\documentclass[a4paper,12pt]{article}%{scrartcl}%
\usepackage{amsmath}
\usepackage{amssymb}
\usepackage{amsthm}
\usepackage{latexsym}
\usepackage{mathrsfs}
\usepackage[dvips]{graphicx}
\usepackage[colorlinks=true, citecolor=blue]{hyperref}
\usepackage{stmaryrd}
\usepackage{txfonts}
\usepackage[T1]{fontenc}
\usepackage{libertine}
\usepackage{pstricks}
\usepackage{color}
\usepackage{tikz}
\usetikzlibrary{decorations.pathreplacing}
\usepgflibrary{shapes.geometric}
\usepackage{slashed}
\usepackage{enumitem}
\usepackage{microtype}
\usepackage[top=1in, bottom=1.5in]%[ includeheadfoot, margin=2.54cm]
{geometry}
\usepackage{longtable}
\usepackage{dsfont}
\usepackage[hang,small]{caption}     
\def\ben{\begin{equation}}
\def\een{\end{equation}}
\def\bena{\begin{eqnarray}}
\def\eena{\end{eqnarray}}
\def\non{\nonumber}

\def\d{{\rm d}}

\def\C{\mathcal{C}}

\def\mr{{\mathbb R}}

\def\F{{\cal F}}

\def\T{{\mathbb T}}

\def\INT{{\mathfrak L}}

\def\c{\operatorname{ch}}
\def\p{\operatorname{pa}}
\def\an{\operatorname{an}}
\def\de{\operatorname{de}}

\def\n{\operatorname{sb}}
\def\ex{8}
\def\nN{\textbf{n}}
\def\root{\mathcal{R}}
\def\bran{\mathcal{B}}
\def\inter{\mathcal{I}}
\def\interoot{\mathcal{I}_\root }
\def\vert{\mathcal{V}}
\def\leaf{\mathcal{L}}

\newcommand{\myid}{{\bf 1}}

\newcommand{\mc}{{\mathbb C}}

\newcommand{\id}{id}

\renewcommand{\O}{{\mathcal O}}

\newcommand{\e}{{\rm e}}
\newcommand{\Pro}{{\mathcal P}}

\def\bra{\langle }
\def\ket{\rangle}

\newcommand{\dext}{\mathfrak{D}_T}

\newtheorem{thm}{Theorem}

\newtheorem{lemma}{Lemma}

\newtheorem{prop}{Proposition}

\newtheorem{corollary}{Corollary}%[section]

\newtheorem{defn}{Definition}%[section]

\title{Associativity of the operator product expansion}
\author{Jan Holland\thanks{\tt jan.holland@uni-leipzig.de}\: and 
Stefan Hollands\thanks{\tt stefan.hollands@uni-leipzig.de}\: 
\\ \\
{\it Institut f\"ur Theoretische Physik, Universit\"at Leipzig
} \\
{\it Br\"uderstr. 16, Leipzig, D-04103, Germany
}  \\
}

\begin{document}
\maketitle
\begin{abstract}

We consider a recursive scheme for defining the coefficients in the operator product expansion (OPE) of
an arbitrary number of composite operators in the context of perturbative, Euclidean quantum field theory
in four dimensions. Our iterative scheme is consistent with previous definitions of OPE coefficients via the
flow equation method, or methods based on Feynman diagrams. It allows us to prove that a strong version of the ``associativity'' condition holds for the OPE to arbitrary orders in perturbation theory. Such a condition was previously proposed in an axiomatic setting in~\cite{Hollands:2008vq} and has interesting conceptual consequences: 1) One can characterise perturbations of quantum field theories abstractly in a sort of ``Hochschild-like'' cohomology setting, 2) one can prove a ``coherence theorem'' analogous to that in an ordinary algebra: The OPE coefficients for a product of two composite operators uniquely determine those for $n$ composite operators. We concretely prove our main results for the Euclidean $\varphi^4_4$ quantum field theory, covering also the massless case. Our methods are rather general, however, and would also apply to other, more involved, theories such as Yang-Mills theories.

\end{abstract}

%\tableofcontents

\section{Introduction}\label{sec:intro}

There exist many different approaches to quantum field theory. Many of these attempt to isolate within quantum field theory a kind of algebraic 
skeleton, which, in a sense depending on the particular framework, defines the theory and dictates its properties. The earliest manifestation of this kind of 
framework is that of  \emph{local quantum physics} due to Haag and Kastler~\cite{haag1992local} which is based on nets of local algebras of operators.
A framework to isolate the algebraic core of many 2-dimensional conformal field theories is the theory of \emph{vertex operator algebras}~\cite{Borcherds1986,Kac1997}. The main idea of this framework is to formalise the properties of the operator product expansion (OPE) in such theories
in order to build an algebraic structure capable of describing many interesting models in two dimensions. 

Since the OPE ought to exist in any local quantum field theory in any dimension~\cite{Wilson:1969ub}, it seems reasonable to define a quantum field theory by it, or more precisely, to attempt to build a self-consistent algebraic structure out of the OPE that can define a quantum field theory. The OPE is the statement that given a complete set of local operators $\O_{A_i}$, and given any sufficiently well-behaved quantum state $\Psi$, one has
\ben\label{OPEintro}
\bra \O_{A_1}(x_1)\cdots \O_{A_N}(x_N)   \ket_{\Psi} \sim \sum_B    \C_{A_1\ldots A_N}^B(x_1,\ldots, x_N) \ \bra \O_B(x_N)  \ket_{\Psi} \, . 
\een
Here, $\C_{A_1\ldots A_N}^B$ are functions (or rather distributions), called \emph{OPE coefficients}, and the symbol $``\sim"$ indicates that the relation is expected to hold asymptotically at short distances, in the sense that the difference between the left and right hand side of \eqref{OPEintro} vanishes if $x_i\to x_N$ for all $i\leq N$. 
In models of perturbative quantum field theory, such as the Euclidean $\varphi_4^4$-theory, the OPE was found to be not only asymptotic, but even convergent, in
the sense that the sum over $B$ in \eqref{OPEintro}  converges even for any finite separation of $x_1,\ldots, x_N$ \cite{Hollands:2011gf, Holland:2014tv}. 

These results strongly suggest that it should indeed be possible to view the OPE coefficients as defining the algebraic skeleton of the theory, and the 1-point 
functions $\bra \O_B(x_N)  \ket_{\Psi}$ as carrying all the information about the state. The theory, then, should be defined by the OPE coefficients, whereas 
specific physical setups should be described by the collection of all 1-point functions, much in the way as a classical field theory is defined by a partial differential equation, and specific physical setups are described by boundary- or initial conditions for determining a given solution. (As an aside, let us point out that this viewpoint is, in fact, not only remarkably close to standard applications of the OPE in deep inelastic scattering, but also very attractive in curved spacetimes~\cite{Wald199411, Hollands:2014uw}, because it is much less clear there what physically preferred states would be in general.)

Of course, in order to define a concrete field theory, one must have a way to determine the OPE coefficients in the first place. The traditional way in Lagrangian field theory is to go back 
to correlation functions and proceed e.g. by the well-known (perturbative) methods described in \cite{Zimmermann1973, Keller:1992by}. This is not really satisfactory if one wants, as we do, to view the OPE coefficients as the primary objects defining the theory, and not Lagrangians or correlation functions. In order to get around this, one clearly needs extra information on the OPE coefficients. One central property (formalised e.g. in the setting~\cite{Hollands:2008vq}) is a kind of \emph{associativity} (also called ``factorisation" or ``consistency") condition, which can be motivated heuristically as follows: Consider an operator product $\O_{A_{1}}(x_{1})\O_{A_{2}}(x_{2})\O_{A_{3}}(x_{3})$, where $x_i\in\mathbb{R}^4$, and assume that  $x_{2}$ is closer to $x_{1}$ than to $x_{3}$, i.e.
\ben\label{introfactcond}
0<\frac{|x_{1}-x_{2}|}{|x_{2}-x_{3}|}\, < 1\, .
\een
Since the OPE is by its very nature a short distance expansion, one may hope to be able to perform the OPE of only the product $\O_{A_{1}}(x_{1})\O_{A_{2}}(x_{2})$ around the point $x_{2}$ first, leaving $\O_{A_3}(x_3)$ as a ``spectator". Such an expansion would have the form
\ben\label{intropartOPE}
\begin{split}
\bra\O_{A_{1}}(x_{1})\O_{A_{2}}(x_{2})\O_{A_{3}}(x_{3}) \ket & \sim  \sum_{B} \C_{A_{1} A_{2}}^{B}(x_{1}, x_{2})  \bra  \O_{B}(x_{2}) \O_{A_{3}}(x_{3})\ket \\
& \sim  \sum_{B,C} \C_{A_{1} A_{2}}^{B}(x_{1}, x_{2})\,  \C_{B A_3}^C(x_2,x_3)   \bra \O_{C}(x_{3})\ket\, ,
\end{split}
\een
where we performed a second OPE in the second line. Comparison with eq.\eqref{OPEintro} yields an associativity condition 
\ben\label{cons1}
\C_{A_1 A_2 A_3}^B(x_1,x_2,x_3)= \sum_{C} \C_{A_{1} A_{2}}^{C}(x_{1}, x_{2})\,  \C_{C A_3}^B(x_2,x_3)  \, .
\een
This condition puts strong restrictions on the OPE coefficients of the theory. To see this, assume also that 
\ben\label{introfactcond2}
0<\frac{|x_{2}-x_{3}|}{|x_{1}-x_{3}|}\, < 1\, .
\een
We can repeat the argument above and arrive at the relation 
\ben\label{cons2}
\C_{A_1 A_2 A_3}^B(x_1,x_2,x_3)= \sum_{C} \C_{A_{2} A_{3}}^{C}(x_{2}, x_{3})\,  \C_{A_1 C}^B(x_1,x_3)  \, .
\een
The requirement of consistency of the alternative expansion schemes \eqref{cons1} and \eqref{cons2} on the domain $0<|x_1-x_2|<|x_2-x_3|<|x_1-x_3|$ yields
\ben\label{consistency0}
 \sum_{C} \C_{A_{1} A_{2}}^{C}(x_{1}, x_{2})\,  \C_{C A_3}^B(x_2,x_3)=  \sum_{C} \C_{A_{2} A_{3}}^{C}(x_{2}, x_{3})\,  \C_{A_1 C}^B(x_1,x_3)\, ,
\een
which encodes highly non-trivial relations between the OPE coefficients.  It was shown in~\cite{Hollands:2008vq} that these have various consequences:
\begin{itemize}
\item  Multipoint OPE coefficients $\C_{A_1\ldots A_N}^B$ are uniquely determined in terms of the two-point coefficients $\C_{A_1A_2}^B$.
\item Deformations (=perturbations) of OPE coefficients can be characterised as a cohomology of Hochschild type.
\item OPE coefficients can be viewed as a (non-conformal, higher dimensional) version of vertex operator algebras. 
\end{itemize}
The formal ``derivation" of the associativity condition presented above is, of course, far from rigorous: For one thing, we have introduced the OPE as an asymptotic expansion, but in  \eqref{introfactcond} and \eqref{introfactcond2} we demanded finite separation of the points $x_1,x_2,x_3$. Furthermore, it is not obvious in what sense, if at all, the \emph{partial} OPE performed in \eqref{intropartOPE} holds. Lastly, we have implicitly exchanged the order of two infinite series in the step from \eqref{intropartOPE} to \eqref{cons1} without any justification. Nevertheless, it is possible to see in some non-trivial examples of field theories such as  in the massless Thirring model~\cite{Olbermann:2012uf}, or in the context of 2 dimensional conformal field theories~\cite{Huang:2007fa} that the strong form of the associativity condition \eqref{consistency0} in fact holds. Unfortunately, the arguments presented in these works are very specific to the peculiar properties of such models, giving no hint whatsoever 
what the situation might be e.g. for perturbatively defined models in Lagrangian field theory.  

In the present paper we show that associativity of the OPE \emph{indeed holds} to all orders in the perturbative Euclidean $\varphi_4^4$-theory. In fact, we even prove a generalisation of eq.\eqref{cons1} to more than three fields:
\begin{thm}\label{thmassoc}
Denote by $[A]$ the dimension of the composite field $\O_A$. At any perturbation order $r\in\mathbb{N}$ in Euclidean $\varphi_4^4$-theory, there exist constants $c,K>0$ such that 
\ben
\begin{split}
\Big|\C_{A_1\ldots A_N}^{B}(x_1,\ldots, x_N)&-\sum_{[C]\leq D} \C_{A_1\ldots A_M}^C(x_1,\ldots, x_M)\, \C_{C A_{M+1}\ldots A_N}^B(x_M,x_{M+1},\ldots, x_N)\Big|_{r\text{-th order}} \\
&
\leq  K  \cdot   \xi^{\frac{D+1}{2}} \cdot 
\left(\frac{D+2}{\sqrt{\xi} -{\xi} }\right)^{c \cdot (\sum_i [A_i]+[B]) } \cdot 
  \frac{\max\limits_{1\leq i\leq N}(\frac{1}{m}, |x_i-x_N| )^{[B]+1}}{\min\limits_{1\leq i<j\leq N}|x_i-x_j|^{\sum_j [A_j]+1}} 
\end{split}
\label{boundspecial}
\een
holds for any $x_1,\ldots, x_N$ such that 
\ben
\label{domain}
\xi := \frac{\max_{1\leq i\leq M}|x_i-x_M|}{\min_{M<j\leq N}|x_j-x_M|} < 1\, ,
\een
where $c=c(r)$ and $K=K(r, A_1,\ldots, A_N,B)$ do not depend on $D$. Since the r.h.s. of \eqref{boundspecial} vanishes in the limit $D\to\infty$, the bound implies that the associativity property 
\ben
\C_{A_1\ldots A_N}^{B}(x_1,\ldots, x_N)=\sum_{[C]\leq D} \C_{A_1\ldots A_M}^C(x_1,\ldots, x_M)\, \C_{C A_{M+1}\ldots A_N}^B(x_M,\ldots, x_N) 
\label{associntro}
\een
holds up to any perturbation order on the domain defined by \eqref{domain}.
\end{thm}
\paragraph{Remark:}
%\begin{enumerate}
%\item 
A much \emph{weaker} version of associativity was previously derived in~\cite{Holland:2012vw}. There, it was shown that 
eq.\eqref{associntro} indeed holds up to any perturbation order, but only on the smaller domain
\ben\label{xi1def}
0<\frac{\max_{1\leq i\leq M}|x_i-x_M|}{\min_{j>M}|x_j-x_M|} < \varepsilon 
\een
for some constant $0<\varepsilon\ll 1$ which moreover decreases with the perturbation order. The weaker version is not suited in order to derive 
\eqref{consistency0}. Furthermore, the weaker version gives the misleading impression that associativity breaks down altogether beyond perturbation 
theory. 
%\end{enumerate}

\vspace{.5cm}

This result suggests that a quantum field theory can be \emph{defined} by a set of OPE coefficients satisfying \eqref{associntro} on the domain \eqref{domain}, together 
with other simple straightforward, and reasonable requirements, see section~\ref{sec:AX} (for more details see \cite{Hollands:2008vq} and also~\cite{Hollands:2008wr, Hollands} for curved spacetimes). 

Even though, thanks to the above theorem, we may now feel much more confident that this viewpoint on QFT is correct, it does not tell us 
how to actually find QFTs, i.e., how to find actual solutions to the consistency requirements \eqref{associntro}. Here a further independent idea is needed. 
This idea is to investigate how, given one solution to the consistency relations (e.g. the Gaussian free field), one can deform this solution to 
another one. As we recall below, one can nicely formulate an abstract deformation (=perturbation) theory of the algebraic structure based on \eqref{associntro} 
wherein perturbations are characterised as elements of some Hochschild type cohomology ring. However, this still does not give a good practical way of 
actually finding perturbations (to all orders in some small parameter, or even finite ones). Instead, we are going to rely on a recently found recursion formula for perturbative OPE coefficients~\cite{Holland:2014ifa}. This recursion formula is derived from the differential equation (a caret $\hat\cdot$ denotes omission)
\ben\label{recursionintro}
\begin{split}
&\partial_{g}\, \C_{A_1\ldots A_N}^B(x_1,\ldots, x_N) = -\int\d^4 y\, \Big[  \C_{\INT A_1\ldots A_N}^B(y,x_1,\ldots, x_N)  \\
&-\sum_{i=1}^N \sum_{[D]\leq [A_i]} \!\!\! \C_{\INT A_i}^D(y,x_i)   \C_{A_1\ldots \widehat{A_i} D\ldots  A_N}^B(x_1,\ldots, x_N)-\!\!\!\! \sum_{[D]< [B]}\!\! \C_{A_1\ldots A_N}^D(x_1,\ldots, x_N)  \C_{\INT D}^B(y,x_N)   \Big] \, ,
\end{split}
\een
for the change of an OPE coefficient if we change the action of the theory by a term of the form $g\O_{\INT}$ (where $g\O_{\INT}$ would be 
$g\varphi^4$ in our model). It is this relation, together with the well-known formulae for the OPE coefficients of the free theory ($g=0$), which is used in this paper to construct the coefficients of the interacting theory order by order in $g$, and to prove theorem \ref{thmassoc}. \emph{The bottom line is that this recursion formula (or the differential equation), together with the consistency relation
\eqref{associntro} completely determine the OPE coefficients of a theory -- hence the theory itself -- 
and that these conditions are mutually consistent with each other.}

This paper is organised as follows: We put our results into the context of axiomatic approaches in section \ref{sec:AX}. Section~\ref{sec:main} contains the main results of the paper, which are then proved for the case of massive fields in section \ref{sec:proof}. The generalisation of the proof to massless fields can be found in section \ref{sec:massless}, followed by our conclusions in section \ref{sec:con}.
 Some technical estimates are moved to an appendix.

\section{General framework for QFT and remarks}

\label{sec:AX}

Before delving into the derivation of the main results of this paper, we would like to explain the wider context provided by a specific proposal for the structure of
QFT~\cite{Hollands:2008vq}. 

\paragraph{OPE algebras:} This framework is intended to formalise the properties of the OPE. In order to avoid writing many indices, one 
associates local fields $O_A$ in the theory with vectors $|v_A\ket $ in some abstract vector space called $V$. 
The space $V$ is assumed to be graded in various ways which reflect the possibility to classify
the different composite quantum fields in the theory by their spin, dimension,
Bose/Fermi character, dimension etc. Thus, for example, if $V_D$ is the space of all fields of a fixed dimension $D$, then 
\ben\label{dualdecomp}
V = \bigoplus_{D } V_D \, .
\een
The infinite sum in this decomposition
is understood without any closure taken. In other words,
a vector $|v\rangle$ in $V$ has only non-zero components in a finite number of the
direct summands in the decomposition~\eqref{dualdecomp}. Typically the set of possible $D$-values is discrete and each $\dim V_D<\infty$\footnote{In order to have a reasonable theory possessing sufficiently many states it is natural to demand a finiteness property of the kind
%
%\ben
$\sum_D q^{-D} \dim V_D <\infty$  for  $0\leq q<1$.
%\een
}.

On the vector space $V$, we assume the existence of an anti-linear, involutive
operation called $\star: V \to V$ which should be thought of as taking the hermitian adjoint of
the quantum fields. We also assume the existence of a linear grading map $\gamma: V \to V$
with the property $\gamma^2 = id$. The vectors corresponding to eigenvalue $+1$ are
to be thought of as "bosonic", while those corresponding to eigenvalue $-1$ are to
be thought of as "fermionic".

\medskip
\noindent

So far, we have only defined a list of objects---in fact a linear space---that we think
of as labelling the various composite
quantum fields of the theory. The dynamical content and quantum nature of
the given theory is next incorporated in the OPE associated with
the quantum fields. This is a hierarchy denoted
\ben\label{hierarchy}
\C = \bigg( \C(-,-), \C(-,-,-), \C(-,-,-,-), \dots \bigg) \, ,
\een
where each $(x_1, \dots, x_N) \mapsto \C(x_1, \dots, x_N)$ is a function on the "configuration space"
\ben
M_N := \{(x_1, \dots, x_N) \in (\mr^4)^N \mid x_i \neq x_j \quad
\text{for all $1 \le i< j \le N$}\} \, ,
\een
taking values in the linear maps
\ben
\C(x_1, \dots, x_N) : V \otimes \cdots \otimes V \to V \, ,
\een
where there are $N$ tensor factors of $V$. (The range of
$\C(x_1, \dots, x_N)$ is actually in the closure $V^{**}$ of $V$ but we do not distinguish this in our notation.)
The components of these maps in a basis of $V$ correspond to the OPE coefficients
mentioned in the previous section. For one point, we set $\C(x_1) = id: V \to V$,
where $id$ is the identity map.

In order to have any chance of imposing stringent consistency conditions of the nature described in section \ref{sec:intro}, the maps $\C(-,\dots,-)$ must be {\em real analytic} functions on $M_N$,
in the sense that their components $\C_{A_1\ldots A_N}^B(x_1,\ldots,x_N):=\bra v_B| \C(x_1,\ldots,x_N) |  v_{A_1}\otimes\ldots \otimes v_{A_N} \ket$ are ordinary real analytic functions on $M_N$ with values in $\mc$.
The basic properties of quantum field theory are then expressed as the following further conditions
on the OPE coefficients:

\medskip
\noindent
\paragraph{\bf C1) Hermitian conjugation:} Denoting by $\star: V \to V$ the
anti-linear map given by the star operation, we have $[\star, \gamma]=0$ and
\ben
\overline{\C(x_1, \dots, x_N)} = \star \, \C(x_1, \dots, x_N) \, \star^{\otimes N}
\een
where $\star^{\otimes N} := \star \otimes \cdots \otimes \star$ is the $N$-fold tensor
product of the map $\star$, and where $\bar\cdot$ denotes complex conjugation.

\medskip
\noindent
\paragraph{\bf C2) Euclidean invariance:} For a suitable representation $R$ 
of ${\rm Spin}(4)$ on $V$ and $a \in \mr^4$, $g \in {\rm Spin}(4)$, we require
\ben
\C(gx_1 + a, \dots, gx_N + a) = R^*(g)
\, \C(x_1, \dots, x_N) \, R(g)^{\otimes N} \, ,
\een
where $R(g)^{\otimes N}$ stands for the $N$-fold tensor product $R(g) \otimes \dots \otimes R(g)$.

\medskip
\noindent
\paragraph{\bf C3) Bosonic nature:} The OPE-coefficients are themselves
"bosonic" in the sense that
\ben
\C(x_1, \dots, x_N) = \gamma \, \C(x_1, \dots, x_N) \, \gamma^{\otimes N}
\een
where $\gamma^{\otimes N}$ is again a shorthand for the $n$-fold tensor product
$\gamma \otimes \dots \otimes \gamma$.

\medskip
\noindent
\paragraph{\bf C4) (Anti-)symmetry:} Let $\tau_{i-1, i} = (i-1 \,\, i)$ be the permutation
exchanging the $(i-1)$-th and the $i$-th object, which we define
to act on $V \otimes \dots \otimes V$ by exchanging the corresponding
tensor factors. Then we have
\bena\label{add1}
&&\C(x_1, \dots, x_{i-1}, x_i, \dots, x_N) \, \tau_{i-1,i} =
\C(x_1, \dots, x_i, x_{i-1}, \dots, x_N) \, (-1)^{F_{i-1}F_i} \\
&& F_i := \frac{1}{2} \,
id^{\otimes(i-1)} \otimes (id-\gamma)
\otimes id^{\otimes(N-i)} \, .
\eena
for all $1<i<N$. Here, the last factor is designed so that bosonic fields
have symmetric OPE coefficients, and fermionic fields have anti-symmetric
OPE-coefficients. The last point $x_N$ and the $N$-th tensor factor
in $V\otimes \dots \otimes V$
do not behave in the same way under permutations, and the formula has to be slightly altered. See~\cite[eq.(3.38)]{Hollands:2008vq} for the corresponding formula. 

\medskip
\noindent
\paragraph{\bf C5) Scaling:} Let ${\rm dim}: V \to V$ be the ``dimension counting operator'', defined to act
by multiplication with $D \in \mr_+$ in each of the subspaces $V_D$  in the decomposition~\eqref{dualdecomp} of $V$ or, put differently, $\dim |v_A\ket =[A] \cdot |v_A\ket$. Then we require that $\myid \in V$ is the unique element up to rescaling
with dimension $\dim(\myid) = 0$, and that $[\dim, \gamma]=0$.

Furthermore, we require that, for any $\delta>0$ and any $(x_1,\ldots,x_n) \in M_n$,
\ben\label{add2}
\lim_{\epsilon\downarrow 0} \epsilon^{[A_1]+\ldots+[A_N]-[B]+\delta } \, \C_{A_1\ldots A_N}^B(\epsilon x_1,\ldots,\epsilon x_N) = 0\, .
\een

\medskip
\noindent
\paragraph{\bf C6) Identity element:} We postulate that there exists a unique element
$\myid$ of $V$ of dimension $[\myid] = 0$,
with the properties $\myid^\star = \myid, \gamma(\myid) = \myid$,
such that
\ben\label{iidop}
\C(x_1, \dots, x_N)|v_1 \otimes \cdots \myid \otimes \cdots v_{N-1} \rangle =
\C(x_1, \dots \widehat{x_i}, \dots x_N) |v_1 \otimes \cdots \otimes v_{N-1}\rangle \, .
\een
where $\myid$ is in the $i$-th tensor position, with $i \le N-1$. When $\myid$
is in the $N$-th tensor position, the analogous requirement takes a slightly
more complicated form (see~\cite[chapter 3]{Hollands:2008vq}).

\medskip
\noindent
\paragraph{\bf C7) Factorisation:} 
\ben
\C(x_1,\ldots, x_N)= \C(x_{M},\ldots, x_N) (\C(x_1,\ldots, x_M) \otimes id^{\otimes(N-M)} )
\een
on the domain
\ben
\frac{\max_{1\leq i\leq M}|x_i-x_M|}{\min_{M<j\leq N}|x_j-x_M|} < 1\, .
\een
Note that this condition is an ``index free" restatement of \eqref{associntro}, the main result of our paper in the context of perturbation theory.

\begin{defn}
A quantum field theory is defined as a pair
consisting of an infinite dimensional vector space $V$ with decomposition~\eqref{dualdecomp}
and maps $\star, \gamma, \dim$ with the
properties described above, together with a hierarchy of OPE coefficients
$\C:= (\C(-,-), \C(-,-,-), \dots)$ satisfying properties C1)--C7).
\end{defn}

It is natural to identify quantum field theories if
they only differ by a redefinition of the fields. Informally, a field redefinition
means that one changes ones definition of the quantum fields of the theory
from $\O_A(x)$ to $\widehat \O_A(x) = \sum_B Z_A^B \O_B(x)$, where $Z_A^B$ is some matrix on field space.
The OPE coefficients of the redefined fields differ from the original ones accordingly by
factors of this matrix. We formalise this in the following definition:

\begin{defn}\label{fieldred}
Let $(V, \C)$ and $(\widehat V, \widehat \C)$ be two quantum field theories. If there exists
an invertible linear map $Z: V \to \widehat V$ with the properties
\ben
Z \, R(g) = \hat R(g) \, Z \,, \quad Z \, \gamma = \hat \gamma \, Z \, , \quad
Z \, \star = \hat \star \, Z \, , \quad Z(\myid) = \widehat \myid \, , \quad
\widehat \dim \, Z \geq Z \, \dim \, ,
\een
together with
\ben\label{redefOPEax}
\C(x_1, \dots, x_N) =  Z^{-1} \, \widehat \C(x_1, \dots, x_N) \, Z^{\otimes N}
\een
for all $N$, where $Z^{\otimes N} = Z \otimes \dots \otimes Z$, then the two quantum field
theories are said to be equivalent, and $Z$ is said to be a field redefinition.
\end{defn}

The main result of our paper, Thm. \ref{thmassoc}, leads to the following conclusion:

\begin{corollary}
The OPE in perturbative Euclidean $\varphi_4^4$-theory satisfies axioms C1)-C7) in the sense of formal perturbation series in $g$, i.e. 
at each fixed order in $g$. 
\end{corollary}

\begin{proof}
The symmetry requirements C1)-C4) and the identity axiom C6) are quite easily checked: They can be explicitly checked in the free theory, and one verifies directly that they are preserved by the recursion formula \eqref{recursionintro}, which we use to define perturbative OPE coefficients. The scaling requirement C5) follows e.g. from the bounds proven in~\cite{Holland:2014ifa}. By far the most non-trivial challenge is to prove C7) (factorisation). This is the content of thm.\ref{thmassoc} of the present paper. 
\end{proof}

\paragraph{Vertex algebras:} Another corollary of theorem \ref{thmassoc} is that perturbation theory defines an analog of a 
vertex operator algebra: First, define \emph{vertex operators} $Y(x,v):V\to V$ as the endomorphism of $V$ whose matrix elements are given by
\ben
\langle v_{C} | Y(x,v_{A}) | v_{B}\rangle:= \C_{AB}^{C}(x,0)
\een
for any $x\neq 0$. The relation \eqref{consistency0}, which is a consequence of our main theorem, may now be written as
\ben\label{vertexas}
Y(v,x)Y(w,y)= Y( Y(v,x-y) w, y )\, ,
\een
where the spacetime arguments are required to satisfy $|x|>|y|>|x-y|>0$ and where $v,w$ are elements of $V$. An almost identical quadratic relation first appeared in the study of conformal field theories in two dimensions, where it is one of the crucial properties (called ``locality condition") of the \emph{vertex operator algebras}~\cite{Kac1997}. It should be stressed, however, that in our context, where conformal symmetry is not required, the condition above is really a highly non-trivial statement on the \emph{convergence} of the infinite sums implicit in eq.\eqref{vertexas}, whereas the same equality in the CFT context is understood in terms of formal power series. 

\paragraph{Abstract perturbation theory:}

The constraint imposed by the factorisation condition C5) at the three point level can be rewritten as
\ben\label{maincondition}
\begin{split}
\C(x_2, x_3)\Big(\C(x_1, x_2) \otimes \id \Big) &- \C(x_1, x_3)\Big(
\id \otimes \C(x_2, x_3) \Big) = 0 \\
&\text{ for } 0<|x_1-x_2|<|x_2-x_3|<|x_1-x_3| ,
\end{split}
\een
which is just an ``index free" version of eq.\eqref{consistency0}. 
Although we will not use this in the present paper, all higher constraints can be derived from this one, see~\cite{Hollands:2008vq}. 
In the very abstract general framework of an OPE algebra, 
we may ask the question when it is possible to find a 1-parameter deformation $\C(x_1, x_2; g)$ of
these coefficients by a parameter $g$ so that the
associativity condition continues to hold, at least in the
sense of formal power series in $g$. (Actually, the analogues of
the symmetry condition~\eqref{add1}, the scaling condition~\eqref{add2}, the
hermitian conjugation, the Euclidean invariance, and the unit axiom should
hold as well for the perturbation. However, these conditions are much more
trivial in nature than~\eqref{maincondition}, because the conditions
are linear in $\C(x_1, x_2)$. These conditions could therefore easily be included in
our discussion, but would distract from the main point.) 

One can show that
such perturbations can be characterised in a cohomological framework. To set up this framework, we consider the non-empty, open domains of $(\mr^4)^N$ defined by
\ben\label{Fndef}
\F_N = \{(x_1, \dots, x_N) \in M_N; \,\,\, r_{1 \, i-1} < r_{i-1 \, i} < r_{i-2 \, i}
< \dots < r_{1i}, \,\,\, 1<i\le N \} \subset M_N \, ,
\een
where $r_{ij}:=|x_i-x_j|$. 
We define $\Omega^N(V)$ to be the set of all real analytic functions
$f_N$ on the domain $\F_N$ that are valued in the linear maps
%~\footnote{The convention
%explained below def.~\ref{deffirst} applies here.}
\ben
f_N(x_1, \dots, x_N): V \otimes \dots \otimes V \to V, \quad (x_1, \dots, x_N) \in \F_N \, .
\een
We next introduce a boundary operator $b: \Omega^N(V) \to \Omega^{N+1}(V)$ by the formula
\bena\label{bfndef}
&&(b f_N)(x_1, \dots, x_{N+1}) :=
\C_0(x_1, x_{N+1})(\id \otimes f_N(x_2, \dots, x_{N+1})) \non\\
&&+ \sum_{i=1}^N (-1)^i f_N(x_1, \dots, \widehat x_i, \dots, x_{N+1})(
\id^{\otimes(i-1)} \otimes \C_0(x_i, x_{i+1}) \otimes \id^{\otimes(N-i)}) \non\\
&&+(-1)^{N+1} \, \C_0(x_N, x_{N+1})(f_N(x_1, \dots, x_N) \otimes \id) \, .
\eena
Here $\C_0(x_1, x_2)$ is the OPE-coefficient of the undeformed theory defined by $g=0$, 
and a caret means omission.
The definition of $b$ involves a composition of $\C_0$ with $f_N$, and
hence, when expressed in a basis of $V$, implicitly involves an
infinite summation over the basis elements of $V$. We must therefore
assume here (and in similar formulas in the following) that 
these sums converge on the set of points $(x_1, \dots, x_{N+1})$ in the domain $\F_{N+1}$.
We shall then say that $bf_N$ exists, and we collect such $f_N$ in the domain of $b$,
\ben
{\rm dom}(b) = \bigoplus_{N \ge 1} \{ f_N \in \Omega^N(V) \mid \text{$bf_N$ exists and is in $\Omega^{N+1}(V)$}\} \, .
\een
When we write $bf_N$, it is understood that $f_N \in \Omega^N(V)$ is
in the domain of $b$. One can show:

\begin{lemma}\label{bsquared}
The map $b$ is a differential, i.e., $b^2f_N = 0$ for
$f_N$ in the domain of $b$ such that $bf_N$ is also in the domain of $b$.
\end{lemma}

Let us define the kernel $Z^N(V, \C)$ of $b$ on $\Omega^N(V)$ as the linear space of
all $f_N \in \Omega^N(V) \cap {\rm dom}(b)$ such that $bf_N = 0$. Similarly,
define the range $B^N(V, \C)$ in $\Omega^N(V)$ to be the linear space of
all $f_N = bf_{N-1}$ such that $f_{N-1} \in \Omega^{N-1}(V) \cap {\rm dom}(b)$ and such
that $f_N$ is in ${\rm dom}(b)$.
By the above lemma, we can then define a cohomology ring associated with the differential
$b$ as
\ben\label{Hndefb}
H^N(V; \C) = \frac{Z^N(V; \C)}{B^N(V; \C)} := \frac{
\{ {\rm ker} \, b : \Omega^N(V) \to \Omega^{N+1}(V)\}
\cap {\rm dom}(b)}{
\{
{\rm ran} \, b: \Omega^{N-1}(V) \to \Omega^N(V)
\}
\cap {\rm dom}(b)
} \, .
\een
As we will now see, the problem of finding a 1-parameter family of
perturbations $\C(x_1, x_2; g)$ such that our associativity
condition~\eqref{maincondition} continues to
hold for $\C(x_1, x_2; g)$ to all orders in $g$ can
be elegantly and compactly  formulated in terms of this ring.
If we let
\ben
\C_i(x_1, x_2) = \frac{1}{i!} \, \frac{d^i}{dg^i} \C(x_1, x_2; g) \Bigg|_{g = 0} \, ,
\een
then we note that the first order associativity condition,
\bena\label{maincondition1}
&&\C_0(x_2, x_3)\Big(\C_1(x_1, x_2) \otimes \id \Big) - \C_0(x_1, x_3)\Big(
\id \otimes \C_1(x_2, x_3) \Big) + \non\\
&&\C_1(x_2, x_3)\Big(\C_0(x_1, x_2) \otimes \id \Big) - \C_1(x_1, x_3)\Big(
\id \otimes \C_0(x_2, x_3) \Big) = 0\,\, ,
\eena
valid for $(x_1, x_2, x_3) \in \F_3$, is equivalent to the statement that
\ben
b \C_1 = 0 \, ,
\een
where here and in the following, $b$ is defined in terms of the unperturbed OPE-coefficient $\C_0$.
Thus, $\C_1$ has to be an element of $Z^2(V; \C_0)$. Let $Z(g): V \to V$
be a $g$-dependent field redefinition in the sense of defn.~\ref{fieldred}, and
suppose that $\C(x_1, x_2)$ and $\C(x_1, x_2; g)$ are connected
by the field redefinition. To first order, this means that
\ben\label{ctrivial}
\C_1(x_1, x_2) = -Z_1 \C_0(x_1, x_2) + \C_0(x_1, x_2)(Z_1 \otimes id + id \otimes Z_1) \, ,
\een
or equivalently, that $bZ_1 = \C_1$,
where $Z_i = \frac{1}{i!} \, \frac{d^i}{dg^i} Z(g) |_{g=0}$.
Thus, the first order deformations of $\C_0$ modulo the trivial ones defined by
eq.~\eqref{ctrivial}
are given by the classes in $H^2(V; \C_0)$. The associativity condition for the $i$-th order
perturbation (assuming that all perturbations up to order $i-1$ exist) can be written as
the following condition for $(x_1, x_2, x_3) \in \F_3$:
\bena\label{bciwi}
&& \C_0(x_2, x_3)\Big( \C_j(x_1, x_2) \otimes \id \Big)
- \C_j(x_1, x_3)\Big(
\id \otimes \C_0(x_2, x_3) \Big) +\\
&& \C_j(x_2, x_3)\Big( \C_0(x_1, x_2) \otimes \id \Big) - \C_0(x_1, x_3)\Big(
\id \otimes \C_j(x_2, x_3)\Big) = w_i(x_1, x_2, x_3) \non \, ,
\eena
where $w_i \in \Omega^3(V)$ is defined by
\ben
w_i(x_1, x_2, x_3) :=
-\sum_{j=1}^{i-1} \C_{i-j}(x_1, x_3)( \id \otimes \C_j(x_2, x_3) ) -
                  \C_{i-j}(x_2, x_3)( \C_j(x_1, x_2) \otimes \id ) \, .
\een
We assume here that all infinite sums implicit in this expression
converge on $\F_3$. This equation may be written alternatively as
\ben\label{bciwiup}
b \C_i = w_i \, .
\een
We would like to define the $i$-th order perturbation by solving
this linear equation for $\C_i$. Clearly, a necessary condition for there to
exist a solution is that $b w_i = 0$ or $w_i \in Z^3(V, \C_0)$, and this can indeed be shown to
be the case. If a solution to eq.~\eqref{bciwiup}
exists, i.e. if $w_i \in B^3(V, \C_0)$, then any other solution will differ from
this one by a solution to the corresponding "homogeneous" equation.
Trivial solutions to the homogeneous equation of the form
$b Z_i$ again correspond to an $i$-th order field redefinition and are not to be
counted as genuine perturbations. In summary, the perturbation series can be continued at $i$-th order
if $[w_i]$ is the trivial class in $H^3(V; \C_0)$,
so $[w_i]$ represents a potential $i$-th
order obstruction to continue the perturbation series.
If there is no obstruction, then the space of non-trivial
$i$-th order perturbations is given by $H^2(V; \C_0)$.
In particular, if we knew e.g. that $H^2(V; \C_0) \neq 0$ while
$H^3(V; \C_0) = 0$, then perturbations could be defined to arbitrary orders in
$g$.

The relationship of this abstract framework with the results of the present paper is the following:
\begin{corollary}
Let $\C_0$ be the OPE coefficients of a free, scalar Euclidean quantum field theory, and 
let $\C_j, j>0$ be their perturbations, as defined by the recursion formula~\eqref{recursionintro}.  Then 
\begin{enumerate}[label=\alph*)]
\item $\C_1$ is a non-trivial element of  $H^2(V; \C_0)$, and 
\item  all higher obstructions $[w_i]\in H^3(V; \C_0) $ vanish.
\end{enumerate}
\end{corollary}
\begin{proof}
Non-triviality of $\C_1$ follows from the fact that the recursion formula \eqref{recursionintro} can not be written as a mere redefinition of the composite fields. The second point, i.e. vanishing of obstructions $[w_i]\in H^3(V; \C_0) $, follows directly from the main result of the present paper, thm.\ref{thmassoc}, because it is equivalent to associativity order-by-order in $g$.
\end{proof}

\section{The Associativity Theorem}\label{sec:main}

In the present section we are going to state our other main results, which will imply the bound stated in thm.~\ref{thmassoc} within perturbative Euclidean $\varphi^4$-theory in four dimensions with classical action
\ben\label{Sclassic}
S=\int\d^4 x  \left(  \frac{1}{2}(\partial\varphi)^2+\frac{m^2}{2} \varphi^2+\frac{g}{4!} \varphi^4 \right)\, .
\een
Throughout the present section we will restrict attention to the massive case $m^2> 0$. The generalisation of our proof to massless fields is discussed afterwards in section \ref{sec:massless}. 

We write the composite operators of our model explicitly as
\ben\label{multiexpl}
\O_{A}=\partial^{\alpha_{1}}\varphi\cdots \partial^{\alpha_{n}}\varphi , \quad  A=(\alpha_{1},\ldots, \alpha_{n})   ,\quad \alpha_{i}\in\mathbb{N}^4\, ,
\een
which means that the corresponding dimension of the field $\O_{A}$ is given by
\ben\label{massdim}
[A]=\sum_{i=1}^{n} (1+|\alpha_{i}|)\,  , \qquad \text{ where } |\alpha|=\sum_{\mu=1}^4 |\alpha_\mu| \text{ for } \alpha\in\mathbb{N}^4\, .
\een
Let us denote the (formal) perturbation series for OPE coefficients by
\ben
\C_{A_1\ldots A_N}^B(x_1,\ldots,x_N)=: \sum_{r=0}^\infty (\C_r)_{A_1\ldots A_N}^B(x_1,\ldots,x_N)\cdot g^r\, ,
\een
where the perturbative OPE coefficients $(\C_r)_{A_1\ldots A_N}^B$ are defined recursively through eq.\eqref{recursionintro}. Further, denote by 
\ben
\begin{split}
&(R_{r}^{D})_{A_{1}\ldots A_{M}; A_{M+1}\ldots A_{N}}^{B}(x_{1},\ldots, x_{N})   :=  \\
&(\C_{r})_{A_{1}\ldots A_{N}}^{B}(x_{1},\ldots,x_{N})-\sum_{s+t= r }\sum_{{[C]\leq D } } (\C_{s})_{A_{1}\ldots A_{M}}^{C}(x_{1},\ldots,x_{M})\, (\C_{t})_{C A_{M+1}\ldots A_{N}}^{B}(x_{M},\ldots,x_{N})
\end{split}
\een
the remainder of the associativity condition at $r$-th perturbation order and truncated at operators $\O_C$ of dimension $[C]=D$. 
Our strategy is to establish the bound \eqref{boundspecial} by an induction which is based on the recursion formula \eqref{recursionintro}. In order to obtain the sharp bound \eqref{boundspecial}, we will have to formulate our induction hypothesis not in terms of the remainder functions $(R_{r}^{D})_{A_{1}\ldots A_{M}; A_{M+1}\ldots A_{N}}^{B}$, but in terms of much more general objects, containing multiple summations over products of OPE coefficients (see definition \ref{defnP} below). These more general expressions are most conveniently organised in terms of decorated rooted trees. Before we can state our main inductive bound, we therefore have to introduce some additional notation.

First, we agree on a vocabulary for rooted trees $T$, which is summarised in the following glossary (cf.~\cite[chapter 3.2.2]{Gross:2013um}): 

\, \, 
\begin{longtable}[l]{p{50pt}  p{350pt}} 
\textbf{Symbol}	&  \textbf{Definition} \\ \hline\hline
$\vert(T)$ & \textbf{Vertices} of the tree $T$.\\
$\leaf(T)$ & \textbf{Leaves} of $T$, i.e. vertices of degree 1 (the degree of a vertex is the number of edges adjacent to it).\\
$\root(T)$ & The \textbf{root} of $T$, $\root\in\vert$. \\
$\inter(T)$ & \textbf{Internal vertices} of $T$, i.e. non-leaf vertices.\\
$\interoot(T)$ & Internal vertices of $T$ without the root, i.e. $\interoot:=\inter\setminus\root$ .\\
$\bran(T)$ & The set of \textbf{branches} of $T$. A branch $b\in\bran$ is a path connecting a leaf to the root, where we use the convention that leaves and root are not part of the branch, i.e. $\bran\subset \interoot$. \\
$\c(v)$ & The \textbf{children} of a vertex $v\in \vert$ are the vertices adjacent to $v$ which are further from the root.  \\
$\p(v)$ & The \textbf{parent} of a vertex $v\in \vert$ is the vertex adjacent to $v$ which is closer to the root. \\
$\n(v)$ & The \textbf{siblings} of a vertex $v\in \vert$ are the children of the parent of $v$ not including $v$ itself, i.e. $\n(v):=\c(\p(v))\setminus\{v\}$.  \\
$\de(v)$ & The \textbf{descendents} of a vertex $v\in \vert$ are the vertices on the paths from $v$ to the leaves.  \\
$\an(v)$ & The \textbf{ancestors} of a vertex $v\in \vert$ are the vertices on the path from $v$ to the root.  \\
\end{longtable}

Next, we add decorations to these trees:
\begin{defn}[Weighted trees]

Let $\vec{x}=(x_1,\ldots, x_N)\in\mathbb{R}^{4N}$ and $\vec{A}=(A_1,\ldots, A_n)$, where $A_i\in\mathbb{N}^{4n_i}$ are multi-indices and where $n\geq N$.  We define $\mathcal{T}(\vec{x};\vec{A})$ to be the set of rooted trees $T$ with the following properties:
\begin{enumerate}
\item $T$ has $n$ vertices and $N$ leaves. 
\item Vertices in $\interoot$ have degree larger than $2$.
\item To each vertex $v\in \vert(T)$ we associate a pair $(x_v,A_v)$ called the \textbf{weight} of $v$, where $x_v\in\mathbb{R}^4$ is a four-vector and $A_v\in\mathbb{N}^{4n_v}$ a multi index, such that
\begin{itemize}
\item  if $v\in\inter(T)$, then $x_v\in \{x_w: w\in \c(v)\}$, i.e. $x_v$ has to be equal to one of the four-vectors associated to the children of $v$. To the leaves $v\in\leaf(T)$ we associate bijectively the vectors $(x_1,\ldots, x_N)$, i.e. $(x_v)_{v\in\leaf}=\vec{x}$.
\item $(A_v)_{v\in \vert(T)}=\vec{A}$, i.e. the mapping between multi-indices and vertices is one-to-one.
\end{itemize}
\end{enumerate}
\end{defn}
\noindent See fig.\ref{fig:scalingtree} for an example of three such trees. 
\begin{figure}[htbp]
\begin{center}
\includegraphics[width=\textwidth]{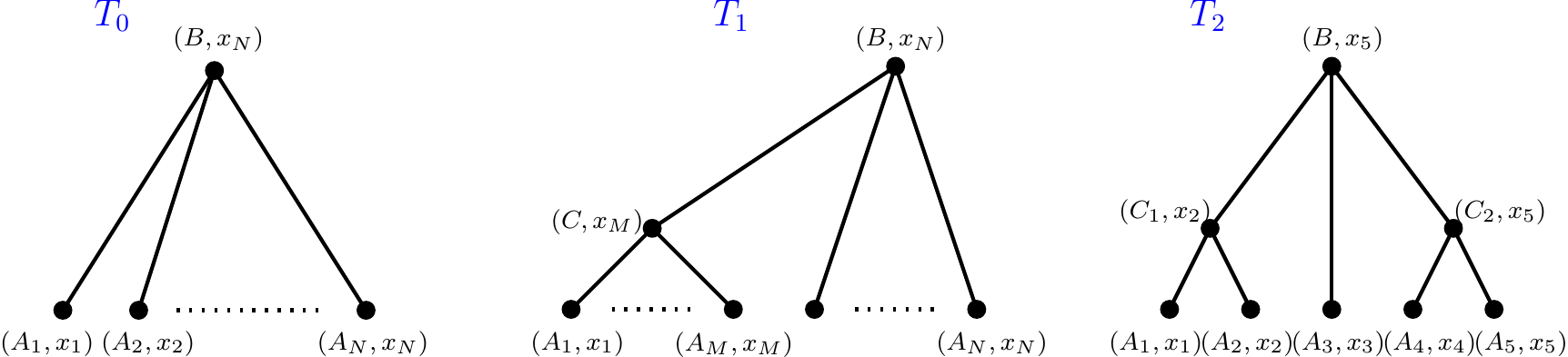}
\end{center}
\caption{Example of weighted trees $T_0\in \mathcal{T}((A_1,\ldots,A_N,B); (x_1,\ldots, x_N))$, $T_1\in \mathcal{T}((A_1,\ldots,A_N,B,C); (x_1,\ldots, x_N))$ and $T_2\in \mathcal{T}((A_1,\ldots, A_5,B,C_1,C_2);(x_1,\ldots, x_5))$.}
\label{fig:scalingtree}
\end{figure}

\noindent 
With this notation in place, we can now give a compact definition of the objects appearing in our induction hypothesis:
\begin{defn}[Contractions of OPE coefficients]\label{defnP}
Given a tree $T\in\mathcal{T}(\vec{x};\vec{A})$, we define
\ben\label{Pdef}
(\Pro_r)(T):=  \prod_{v\in \inter(T)} \sum_{\sum\limits_{u\in \inter(T)}r_u=r}    (\C_{r_v})_{(A_e)_{e\in \c(v)}}^{A_{v}  }((x_{e})_{e\in \c(v)}; x_v)\, .
\een
The argument behind the semicolon in the OPE coefficients specifies the reference point, i.e.
\ben
\C_{A_1\ldots A_N}^B(x_1,\ldots, x_N; x_1)=\C_{A_2\ldots A_N A_1}^B(x_2,\ldots, x_N, x_1)
\een
for example. 
\end{defn}

\paragraph{Examples:} For the weighted trees depicted in fig.\ref{fig:scalingtree}, the definition yields
\begin{align}
(\Pro_r)(T_0)&= (\C_{r})_{A_1 A_2 \ldots A_N}^{B}(x_1,x_2,\ldots,x_N) \label{ex0} \\
\label{ex1}
(\Pro_r)(T_1)&= \sum_{ r_1+r_2=r  }(\C_{r_1})_{A_1 \ldots A_M}^{C}(x_1,\ldots,x_M) \ (\C_{r_2})_{C A_{M+1}\ldots A_N}^{B}(x_M,\ldots, x_N) \\
(\Pro_r)(T_2)&= \sum_{r_1+r_2+r_3=r}(\C_{r_1})_{A_1 A_2}^{C_1}(x_1,x_2) \ (\C_{r_2})_{A_4 A_5}^{C_2}(x_4,x_5) \ (\C_{r_3})_{C_1 A_3 C_2}^{B}(x_2,x_3,x_5)  \, .
\end{align}
We are now ready to state our second main theorem, which directly implies theorem \ref{thmassoc}.
\begin{thm}\label{mainthm}
Up to any perturbation order $r\in\mathbb{N}$, OPE coefficients of massive Euclidean $\varphi_4^4$-theory satisfy the following two properties:
\begin{enumerate}[label=(\alph*)]
\item\label{a}  Given a tree $T\in\mathcal{T}(\vec{x};\vec{A})$ and given a collection of integers $(D_i)_{i\in\interoot}$, fix any branch $b\in\bran(T)$ in the tree such that\footnote{Such a branch  exists for every tree $T$. In fact, it is not hard to see that the number of such branches is equal to the degree of the root vertex $\root$ of $T$.} $x_v=x_w$ for all $v,w\in b(T)$ and  define the shorthand $\dext:=\sum_{v\in\leaf\cup\root(T)}[A_v]$. 
 For any choice of constants $\varepsilon\in (0,2^{-(\dext+4r+3)}]$ and $\delta_v\in (0,1)$, one has the bound
\ben\label{hypothesis1}
\begin{split}
\Big|\prod_{i\in\interoot(T)} \sum_{[A_i]=D_i} (\Pro_r)(T) \Big| \leq& \  \frac{ \max\limits_{i \in \c(\root)}|x_i-x_{\root}|^{[A_\root]}   }{ \prod\limits_{i\in \leaf} \min\limits_{j\in \n(i)}|x_i-x_j|^{[A_i]}  }  \ \prod\limits_{i\in\interoot} \xi_i^{D_i}      \\
\times& \ K\  \bigg( \frac{(1+\varepsilon)^{\sum_{w\in b}[A_w]}}{\varepsilon^{\sum_{v\in\vert\setminus b} [A_v]}} \cdot \prod_{w\in b} (D_w+1)^{\dext} \bigg)^{\ex^{r+1}}     \\
\times&  \prod\limits_{v\in \inter} \left(\frac{\sup_{i\in \c(v)}(|x_i-x_{v}|, 1/m)}{  m^{\theta(\Delta_{v})}  \min\limits_{i\neq j\in \c(v)}|x_i-x_j|^{1+ \theta(\Delta_{v})}} \right)^{\delta_v}     
\end{split}
\een
where the constant $K>0$ depends neither on the integers $D_i$ nor on $\varepsilon$ or the $\delta_v$, where 
\ben
\xi_i(T):=\frac{\max_{e\in \c(i)}|x_e-x_{i}|}{\min_{e\in \n(i)}|x_e-x_{i}|}\, ,
\een
where $\theta$ is the Heaviside step function\footnote{We use the convention $\theta(0)=0$.} and where $\Delta_v:=\sum_{e\in \c(v)}[A_e]-[A_{v}]$. 

\item For any $\xi<1$ one has 
\ben\label{hyp2}
\lim_{D\to\infty}(R_{r}^{D})_{A_{1}\ldots A_{M}; A_{M+1}\ldots A_{N}}^{B}(x_{1},\ldots, x_{N})=0\, ,
\een
where $\xi$ is defined as in eq.\eqref{xi1def}.
\end{enumerate}
\end{thm}
\paragraph{Remark:}
Before we come to the proof of the theorem, let us take a moment to have a closer look at the result in order to get a better intuition for the complicated expression~\eqref{hypothesis1}. The origin of the various terms in the bound \eqref{hypothesis1} can be roughly understood as follows: 
\begin{enumerate}
\item The first line reflects the behaviour one would expect from naive power counting if one assumes that an OPE coefficient behaves as $\C_{A_1\ldots A_N}^B(x_1,\ldots,x_N) \sim \max|x_i-x_N|^{[B]}/\prod_i\min|x_i-x_j|^{\sum [A_i]}$. 
\item The second line captures all combinatorial factors, in particular those caused by the summations over multi indices $[C_i]=D_i$ associated to the internal vertices of the tree $T$ and those arising  in perturbation theory. Note that only this second line depends on the perturbation order $r$. 
\item In the last line, the factors including the Heaviside function are a relict of the exponential decay of the massive propagator. These factors are needed in the induction in order to avoid infrared divergences. Finally, the factor $(\sup_{i\in \c(v)}(|x_i-x_{v}|, 1/m)/\min\limits_{i\neq j \in \c(v))}|x_i-x_j|)^{\delta_i}$ in the last line reflects the fact that naive power counting only holds up to logarithms once we proceed to higher orders in perturbation theory. We note also that the bound diverges if we set the mass $m$ to zero.
\end{enumerate}

\vspace{.3cm}

\begin{proof}[Proof of theorem \ref{thmassoc}]
As mentioned in the introduction, theorem \ref{thmassoc} can be derived straightforwardly from theorem \ref{mainthm}. To see this, note that eq.\eqref{hyp2} implies
\ben\label{thmtocor}
\begin{split}
&(R_{r}^{D})_{A_{1}\ldots A_{M}; A_{M+1}\ldots A_{N}}^{B}(x_{1},\ldots, x_{N})\\
&=\sum_{ {r_1+r_2=r} } \sum_{[C]>D}(\C_{r_1})_{A_1 \ldots A_M}^{C}(x_1,\ldots,x_M) \ (\C_{r_2})_{C A_{M+1}\ldots A_N}^{B}(x_M,\ldots, x_N)=\sum_{[C]>D}(\Pro_r)(T_1)
\end{split}
\een
where $T_1 \in \mathcal{T}((A_1,\ldots,A_N,B,C); (x_1,\ldots, x_N))$ is the tree depicted in fig.\ref{fig:scalingtree}. We can now use the bound \eqref{hypothesis1} to estimate the right hand side. The infinite sum can be bounded using the inequality 
\ben\label{dsumineq}
\sum_{d>D}  \left(\xi (1+\varepsilon)^{\ex^{r+1}}\right)^d (d+1)^{\ex^{r+1}\dext} \leq \left(\xi (1+\varepsilon)^{\ex^{r+1}}\right)^{D+1}  \left(\frac{D+2}{1-\xi (1+\varepsilon)^{\ex^{r+1}}}\right)^{\ex^{r+1}\dext+1} \!\!\!  (\ex^{r+1}\dext)! \, ,
\een
where $\dext=\sum_{i=1}^N [A_i]+[B] $ and where we chose $\varepsilon$ small enough such that $(1+\varepsilon)^{\ex^{r+1}}\xi <1$. In particular, we are free to choose for example $(1+\varepsilon)^{\ex^{r+1}}  = 1/\sqrt{\xi}$. After simple algebraic manipulation and absorbing some factors into the constant $K$, we arrive at \eqref{boundspecial}.
\end{proof}

\vspace{.5cm}

The reader may wonder at this stage why we derive the rather complicated bounds \eqref{hypothesis1} on the objects $(\Pro_r)(T)$ if we are eventually only interested in the simpler bound \eqref{boundspecial}. The reason for this apparent detour lies in the fact that the bound \eqref{boundspecial} itself is not suited for the induction we are using. Roughly speaking, the main technical problem with an induction based on the remainder $(R_{r}^{D})_{A_{1}\ldots A_{M}; A_{M+1}\ldots A_{N}}^{B}$ comes from the fact that one wants to avoid making relatively rough estimates for the summations over multi-indices appearing in the recursion formula \eqref{recursionintro}. As an example, one would have to use estimates like
\ben\label{weakest}
\Big|\sum_{[C]\leq D}   (\C_s)_{A_1\ldots A_M}^C (R_{t}^{D})_{\varphi^4 C; A_{M+1}\ldots A_{N}}^{B}  \Big| \leq  \sum_{[C]\leq D} \Big| (\C_s)_{A_1\ldots A_M}^C\Big| \cdot \Big|(R_{t}^{D})_{\varphi^4 C; A_{M+1}\ldots A_{N}}^{B}\Big| \, .
\een
As it turns out, such estimates lead to unwanted combinatoric factors of the form $c^D$ for some constant $c>1$, which accumulate for every iteration of the recursion formula. As a result, one is led to an associativity condition that gets weaker as the perturbation order increases (similar to the result derived in~\cite{Holland:2012vw}, see also the remark below theorem~\ref{thmassoc}). Our solution to this problem is to estimate the objects $\prod_{i\in\interoot(T)} \sum_{[A_i]=D_i}(\Pro_r)(T)$, which include multiple sums over multi-indices $A_i$ and which thereby allow us to avoid weak estimates of the type \eqref{weakest}, i.e. we never have to ``pull the modulus inside the sum". The formulation in terms of rooted trees is further convenient in order to keep track of the various terms generated by the recursion formula \eqref{recursionintro}, and in particular in order to verify cancellations of divergent terms in the recursion as discussed in more detail in the next section.

\section{Proof of theorem \ref{mainthm}}\label{sec:proof}

In the present section we are going to present the proof of theorem \ref{mainthm}, which proceeds by induction in the perturbation order $r$. Before we get to the details of this rather long line of arguments, let us give a brief overview of the general strategy and the main steps followed in this section.

\begin{description}

\item[Induction start (sec. \ref{sec:start}):] Theorem \ref{mainthm} makes two claims, namely the bound \eqref{hypothesis1} and the convergence property \eqref{hyp2}. Thus, our aim is to prove both these properties for $r=0$, i.e. within the free theory. In this simple case, we can treat the problem explicitly using mainly Wick's Theorem. Namely, we can write down an explicit representation for the zeroth order OPE coefficients [see eq.\eqref{OPEpf}], and we then generalise this representation to the objects of interest $ (\Pro_r)(T)$ [see lemma \ref{lemHf}]. With this representation at hand, we can a) derive the claimed bounds \eqref{hypothesis1} [see subsection \ref{subsec:startA}] and we can b) check for convergence of the associativity condition [see subsection \ref{subsec:startB}]. 

\item[Induction step (sec. \ref{sec:step}):] Our aim is again to prove the bound \eqref{hypothesis1} and the convergence property \eqref{hyp2}, but now at perturbation order $r+1$, under the assumption that both these properties hold up to order $r$. Our main ingredient here is the recursion formula \eqref{recursionintro}, which implies a corresponding recursion formula for the objects of interest $ (\Pro_{r})(T)$ [see eq.\eqref{remr+1}]. This formula allows us to establish bounds on $ (\Pro_{r+1})(T)$ in terms of an integral over objects at order $r$, for which we can use the inductive bound by assumption [see subsection \ref{subsec:stepA}]. In order to verify the bound \eqref{hypothesis1} at order $r+1$, it then remains to estimate this integral. 

Here some care has to be taken, since the individual terms under the integral generated by the recursion formula are in fact divergent. One has to make use of cancellations between such terms in the potentially dangerous integration regions, which can be nicely organised with the help of our tree notation. Thus, we decompose $\mathbb{R}^4$ into various intermediate-, short- and large-distance regions, and we consider the integral over these regions separately. The cancellations between divergent terms then can be seen to follow from the associativity condition \eqref{hyp2} at order $r$, and the bound \eqref{hypothesis1} can be verified in each region by rather straightforward computations. 

Finally, to prove the convergence property \eqref{hyp2} at order $r+1$, we once again use the recursion formula in order to express the associativity remainder at order $r+1$ in terms of an integral over quantities at order $r$. Then, using  the bound \eqref{hypothesis1} at order $r+1$ that we have just verified, we can exchange the order of the integral with the limit $D\to\infty$, which leads to a vanishing integrand, and thus to a vanishing remainder as claimed [see subsection~\ref{subsec:stepB}]. 

\end{description}

\subsection{Induction start: The free theory}\label{sec:start}

Our aim in this section is to verify the two hypotheses of theorem \ref{mainthm}, i.e. the bound \eqref{hypothesis1} and the convergence property \eqref{hyp2}, for free quantum fields. This will be achieved by giving an explicit representation for the objects $(\Pro_0)(T)$, which is obtained with the help of Wick's Theorem. 

To derive this representation, let us start with the simplest possible trees, i.e. let $T_0\in\mathcal{T}(\vec{x};\vec{A})$ be any tree whose only internal vertex is the root (such as $T_0$ in fig.\ref{fig:scalingtree}). Recall from our example in eq.\eqref{ex0} that the corresponding expression $(\Pro_0)(T_0)$ is simply a single OPE coefficient. For concreteness, we write the multi indices $A_v\in \mathbb{N}^{4n_v}$ associated to the vertices $v\in \vert(T)$ explicitly as, 
\ben\label{compopdef}
A_v=(\alpha_{v,1},\ldots, \alpha_{v,n_v})\quad, \O_{A_v}=\partial^{\alpha_{v,1}}\varphi\cdots \partial^{\alpha_{v,n_v}}\varphi\quad  , \alpha_{v,i}\in\mathbb{N}^4\, .
\een
Wick's Theorem then implies the convenient representation (this follows from the standard definition of OPE coefficients for a free scalar field, see e.g.~\cite[eq. (2.56)]{Holland:2014ifa})\footnote{The r.h.s. of \eqref{OPEpf} is also called the \emph{Hafnian} of the matrix $\Sigma$, see~\cite{Caianiello:1973wt}.}
 \ben\label{OPEpf}
(\Pro_0)(T_0)= (\C_0)_{(A_v)_{v\in \leaf(T_0)}}^{A_\root}((x_v)_{v\in \leaf(T_0)};  x_{\root(T_0)})= 
\sum_{\sigma\in \mathfrak{M}(\vert(T_0))}\prod_{[(v,i),(w,j)]\in\sigma} \Sigma_{(v,i),(w,j)}
\een
where the $\sum_{v\in\vert}n_v \times \sum_{v\in\vert}n_v$-matrix ${\Sigma}$ is defined as 
\ben\label{sigmadef}
{\Sigma}_{(v,i)(w,j)}:= \begin{cases}
\partial_{x_v}^{\alpha_{v,i}}\partial_{x_w}^{\alpha_{w,j}}\Delta(x_v-x_w) & \text{ for } v,w\neq \root, v\neq w  \\
 \frac{\partial^{\alpha_{v,i}}_{x_{v}} (x_v-x_\root)^{\alpha_{w,j}}}{\alpha_{w,j}!}
& \text{ for } v\neq \root, w=\root \\
0 & \text{ for } v=w
\, ,
\end{cases}
\een
where 
\ben\label{propagator}
\Delta(x) :=  \frac{1}{(2\pi)^2}
 \int \frac{\e^{i p x}}{p^2+m^2}\, \d^4 p 
\een
is the Euclidean propagator, and where $\mathfrak{M}(\vert(T_0))$ is the set of \emph{perfect matchings} on the vertices $\{(v,i)_{v\in\vert(T_0)}^{i\in\{1,\ldots, n_v\}}\}$. A perfect matching on a vertex set $I$ is a set of edges  such that each vertex in $I$ is incident to exactly one edge. Figure \ref{fig:matching} illustrates in a simple example how to obtain the r.h.s. of eq.\eqref{OPEpf} from a given tree $T_0$.
\begin{figure}[htbp]
\begin{center}
\includegraphics{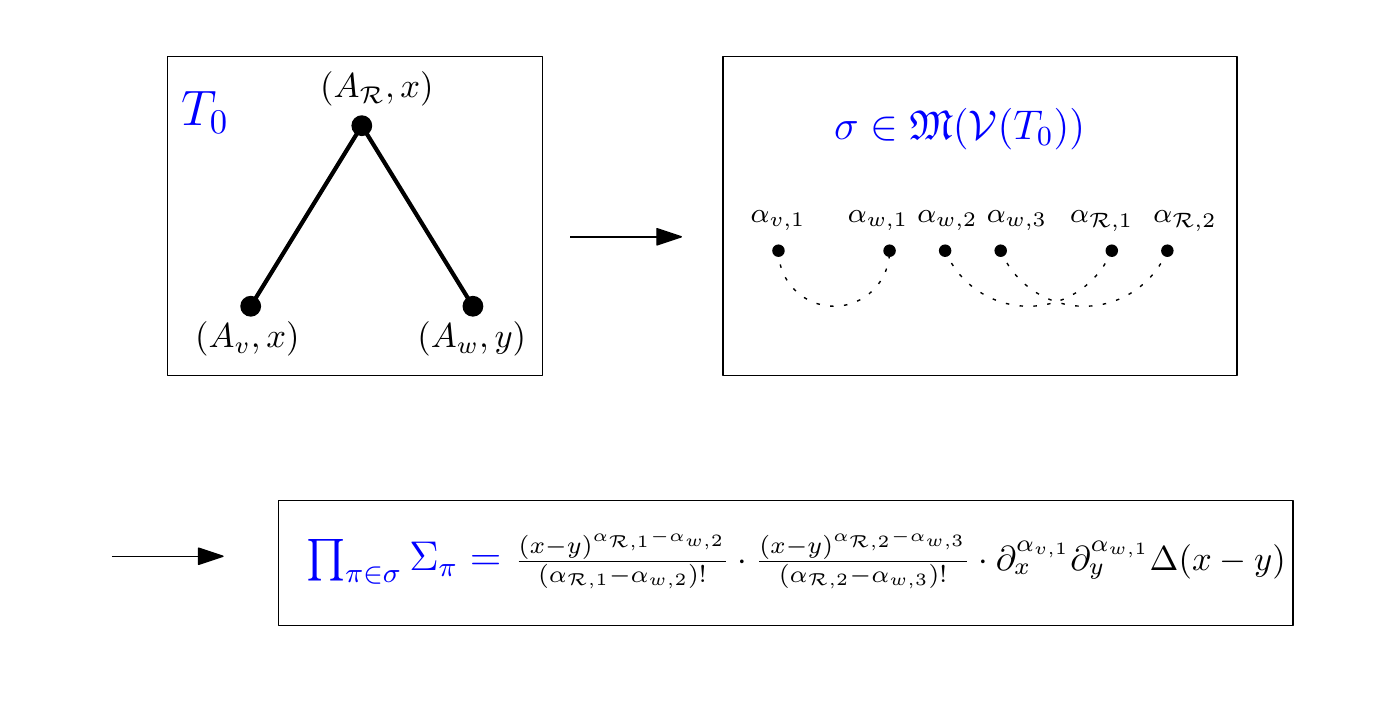}
\end{center}
\caption{\emph{From trees to OPE coefficients:} Given a tree $T_0$ [top left] with multi-index labels $A_v=(\alpha_{v,1}), A_w=(\alpha_{w,1},\alpha_{w,2},\alpha_{w,3})$ and $A_\root=(\alpha_{\root,1},\alpha_{\root,2})$, we obtain perfect matchings $\sigma \in \mathfrak{M}(\vert(T_0))$ [top right] by decomposing the indices $(\alpha_{v,1},\alpha_{w,1},\alpha_{w,2},\alpha_{w,3},\alpha_{\root,1},\alpha_{\root,2})$ into pairs. The contribution to the OPE coefficient  $(\Pro_0)(T_0)=(\C_0)_{A_w A_v}^{A_\root}(y,x)$ corresponding to a perfect matching $\sigma$ is given according to eq.\eqref{sigmadef} by explicit expressions involving the propagator $\Delta$ [bottom].}
\label{fig:matching}
\end{figure}

We now want to extend this representation to more complicated trees $T$.  As a warm up, let us first consider trees $T_1$ with only one internal vertex $u$ besides the root, $\interoot=\{u\}$, such as the tree displayed in fig.\ref{fig:scalingtree}. As mentioned earlier in \eqref{ex1}, trees of this type correspond to a product of two OPE coefficients, $(\C_0)_{(A_v)_{v\in\c(u)}}^{A_{u}}(\C_0)_{(A_v)_{v\in\c(\root)}}^{A_{\root}}$. Using the representation \eqref{OPEpf}, we can express this product in terms of two weighted perfect matchings: 
\ben\label{pT1prod}
\sum_{[A_u]=D}(\Pro_0)(T_1)=  \sum_{[A_u]=D} \Big(\sum_{\sigma_1\in \mathfrak{M}(\vert(T_1^1))}\prod_{\pi_1\in\sigma_1} \Sigma_{\pi_1} \Big) \, \cdot\, \Big( \sum_{\sigma_2\in \mathfrak{M}(\vert(T_1^2))}\prod_{\pi_2\in\sigma_2} \Sigma_{\pi_2}\Big)
\een
Here we write $T_1^a$ for the tree which is obtained from $T_1$ by deleting all vertices and edges above the internal vertex $u\in\interoot$, and $T_1^b$ for the tree which results from $T_1$ by deleting all vertices and edges below $u\in\interoot$, see fig.\ref{fig:subtrees}. 
\begin{figure}[htbp]
\begin{center}
\includegraphics{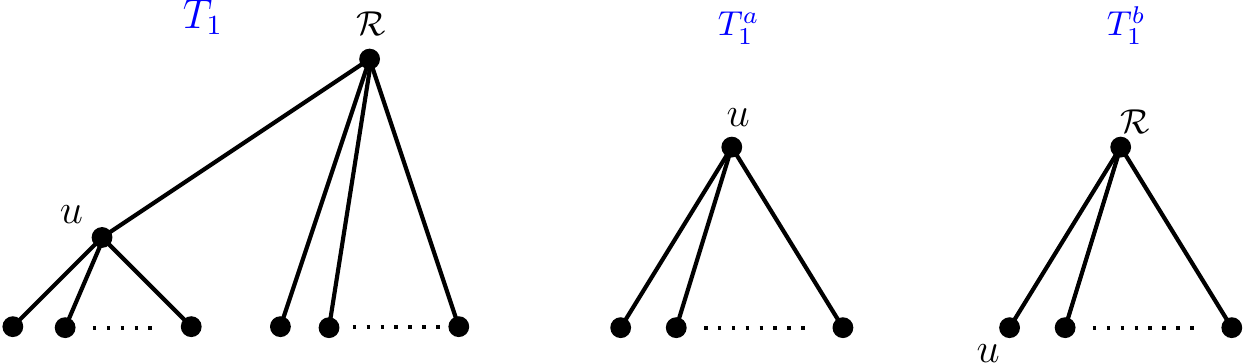}
\end{center}
\caption{The decomposition of a tree $T_1$ into subtrees $T_1^a,T_1^b$ without internal vertices.}
\label{fig:subtrees}
\end{figure}

\noindent Equation \eqref{pT1prod} can be simplified in various ways. Firstly, we note that the internal vertex $u\in\interoot$ appears in both matchings, which we can highlight by writing the above equation as follows:
\ben\label{pT1uhigh}
\sum_{[A_u]=D}(\Pro_0)(T_1)= \sum_{\substack{\sigma_1\in \mathfrak{M}(\vert(T_1^a)) \\   \sigma_2\in \mathfrak{M}(\vert(T_1^b)) } } \prod_{\substack{ [(v,i),(w,j)] \in \sigma_1\cup \sigma_2 \\ v,w\neq u  }}   \hspace{-.3cm} \Sigma_{(v,i),(w,j)}  \sum_{[A_u]=D}  \prod_{\substack{ [(v,i),(u,k)]\in\sigma_1 \\ [(u,k),(w,j)]\in\sigma_2  }}  \hspace{-.2cm} \Sigma_{(v,i),(u,k)}\cdot  \Sigma_{(u,k),(w,j)}
\een
The product on the very right, which contains all matchings involving the internal $u$-vertex, can then be written as
\ben\label{mergeweight}
\begin{split}
&\sum_{|\alpha_{u,k}|=d}\Sigma_{(v,i)(u,k)}\cdot \Sigma_{(u,k)(w,j)} \\
&=\begin{cases}
 \sum\limits_{|\alpha_{u,k}|=d}\frac{\partial^{\alpha_{v,i}}_{x_{v}}(x_v-x_u)^{\alpha_{u,k}}}{\alpha_{u,k}!} \, \partial^{\alpha_{u,k}}_{x_u}\partial^{\alpha_{w,j}}_{x_w}\Delta(x_u-x_w) =  \T_{x_v\to x_u}^{d-|\alpha_{v,i}|} \Sigma_{(v,i)(w,j)} &  w\neq \root \\
   \sum\limits_{|\alpha_{u,k}|=d}\frac{\partial^{\alpha_{v,i}}_{x_{v}}(x_v-x_u)^{\alpha_{u,k}}}{\alpha_{u,k}!}  \, \frac{\partial^{\alpha_{u,k}}_{x_{u}}(x_u-x_\root)^{\alpha_{w,j}}}{\alpha_{w,j}!} =  \T_{x_v\to x_u}^{d-|\alpha_{v,i}|} \Sigma_{(v,i)(w,j)} &  w=\root
\end{cases}
\end{split}
\een
 where $\T^d$ is the Taylor expansion operator 
\ben
\T_{x\to y}^{d}f(x):=\begin{cases}
 \sum_{|v|=d} \frac{(x-y)^{v}}{v!}\, \partial^{v}_{y}f(y) & \text{ for } d\geq 0 \\
 0 & \text{ for }d<0\, .
 \end{cases}
\een
We can further simplify eq.\eqref{pT1uhigh} by expressing the summation over the matchings $\sigma_1,\sigma_2$ in terms of matchings $\sigma\in \mathfrak{M}(\leaf\cup\root(T_1))$. This is achieved by \emph{merging} the two matchings at the $(u,k)$ vertices, as shown in fig.\ref{fig:merge}. 
\begin{figure}[htbp]
\begin{center}
\includegraphics{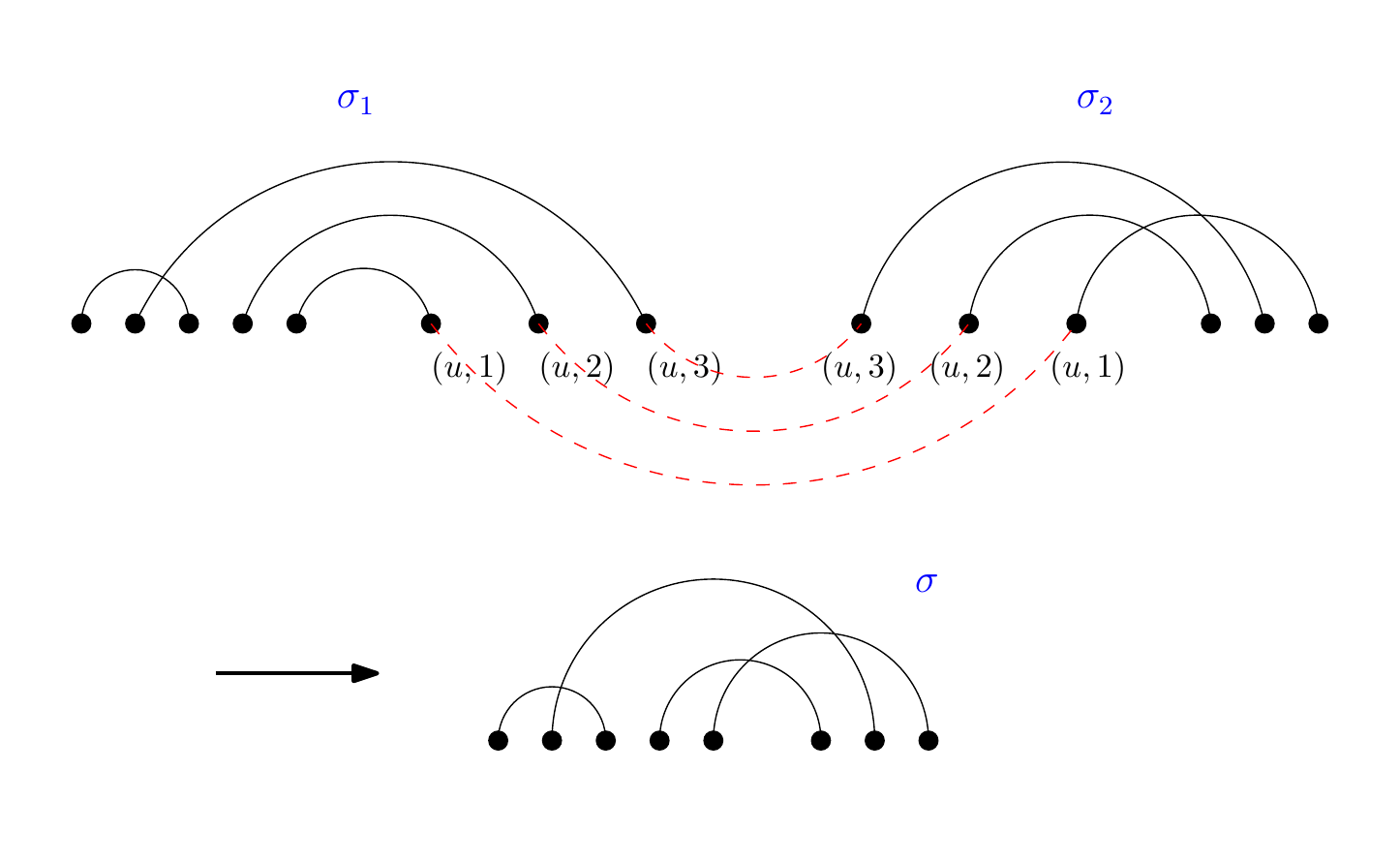}
\end{center}
\caption{Given two perfect matchings $\sigma_1\in\mathfrak{M}(\vert(T_1^a)),\sigma_2\in\mathfrak{M}(\vert(T_1^b))$ [top] we obtain a matching $\sigma\in \mathfrak{M}(\leaf\cup\root(T_1))$ [bottom] by merging the vertices corresponding to the internal vertex $u\in\interoot(T_1)$ as indicated by the dashed red lines.}
\label{fig:merge}
\end{figure}
Note however that this mapping $(\sigma_1,\sigma_2)\to \sigma$ is not one to one: Exchanging two vertices $(u,k)$ and $(u,k')$ yields the same matching $\sigma$. For a given $\sigma$, we therefore pick up a symmetry factor $|I(\sigma)|!$, where (recall that $\de(u)$ denotes the descendants of $u$)
\ben
I(\sigma):= \{ [(v,i)(w,j)]\in \sigma\, :\,  v\in\de(u), w\notin\de(u)    \}
\een
is the set of merged edges, i.e. those adjacent to a $u$-vertex in the original matchings $\sigma_1,\sigma_2$. The matching procedure thus leads to the formula
\ben\label{Mspecial}
\sum_{[A_u]= D} (\Pro_{0})(T_1)= \hspace{-.4cm}  \sum_{\sigma\in \mathfrak{M}(\leaf\cup\root(T_1))}\hspace{-.6cm} |I(\sigma)|!  \sum_{\vec{d}\in\mathcal{D}(\sigma)} \prod_{\pi=[(v,i)(w,j)]\in\sigma} 
\begin{cases}
 \T_{x_{v}\to x_{u}}^{d_{\pi}-|\alpha_{v,i}|}   \Sigma_{\pi} & \text{ if }\pi\in I(\sigma), v\in\c(u) \\
 \Sigma_{\pi} & \text{ if }\pi\notin I(\sigma)
\end{cases}
\een
where we summarised the possible assignments of the Taylor expansion degrees to the merged lines in the definition
\ben
\mathcal{D}(\sigma)=\{ \vec{d}=(d_{\pi})_{\pi\in I(\sigma)}\, :\, d_{\pi}\in\mathbb{N},   \sum_{\pi\in I(\sigma)} (d_{\pi}+1) =D  \}\, .
\een
We can generalise this strategy to the expression $(\Pro_0)(T)$ for more complicated trees $T$. For this purpose, let us first define the sets
\ben
{I}_{u}(\sigma):= \{ [(v,i)(w,j)]\in\sigma \, :\,  v\in \de(u), w\notin \de(u)  \}
\een
for any $u\in\interoot$, which contain all edges which are merged by connecting two $u$-vertices. Further, define
\ben
\mathcal{D}(\sigma)=\{\vec{d}= (d^u_{\pi})^{u\in\interoot}_{\pi\in I_{u}(\sigma)}\, :\, d^u_{\pi}\in\mathbb{N},      \sum_{\pi\in I_{u}(\sigma)} (d^u_{\pi}+1) =D_u  \}\, ,
\een
which is the set of all assignments of the Taylor expansion degrees to the merged edges.  We then have the compact formula:
\begin{lemma}\label{lemHf}
Let $T\in \mathcal{T}(\vec{x};\vec{A})$. Then
\ben\label{PT0}
\prod_{u\in\interoot(T)} \sum_{[A_u]=D_u} (\Pro_0)(T)=   \sum_{\sigma\in \mathfrak{M}(\leaf\cup\root(T))} \prod_{u\in\interoot(T)}|I_u(\sigma)|!  \sum_{\vec{d}\in\mathcal{D}(\sigma)} \prod_{\pi\in\sigma}  M_{\pi}(\vec{d})%(T,\vec{D},\sigma)
\een
where the matrix $M(\vec{d})$ is given by (recall that by $\an(v)$ we denote ancestors of $v$)
\ben\label{Mweight}
M_{\pi=[(v,i)(w,j)]}(\vec{d})=  \prod_{u\in \an(v)\setminus\an(w)}\T_{x_{v}\to x_{u}}^{d_{\pi}^{u}-|\alpha_{v,i}|} \,  \prod_{s\in \an(w)\setminus\an(v)}\T_{x_{w}\to x_{s}}^{d_{\pi}^{s}-|\alpha_{w,j}|}
\quad \Sigma_{\pi}
\een
 The product over the vertices  $u,s$ in eq.\eqref{Mweight} is ordered from leaf to root, i.e. every vertex is to the left of its ancestors. 
\end{lemma}
\begin{proof}
The proof works by induction in the number of internal vertices $|\interoot(T)|$. In the simple examples above we have already dealt with the cases $|\interoot(T_0)|=0$ and $|\interoot(T_1)|=1$, so the induction start has already been taken care of. The induction step works as follows: Assuming that lemma \ref{lemHf} holds for all trees $T'\in \mathcal{T}(\vec{x};\vec{A})$ with up to $n$ 
 internal vertices, $|\interoot(T')|\leq n$, we have to show that the lemma also holds for trees $T$ with $n+1$ internal vertices, i.e. for $|\interoot(T)|= n+1$. 

The idea of the proof is analogous to the simple example with one internal vertex discussed above: Fix any internal vertex $u\in\interoot(T)$ and denote by $T^a$ the tree obtained from $T$ by deleting all vertices and edges above the vertex $u$, and by $T^b$ the tree obtained from $T$ by deleting all vertices and edges below $u$. Since both $T^a$ and $T^b$ have at most $n$ internal vertices, we can use the induction hypothesis in order to express $(\Pro_0)(T)$ as a product of the form
\ben
\begin{split}
\prod_{s\in\interoot(T)}\sum_{[A_s]=D_s}(\Pro_0)(T)=\sum_{[A_u]=D_u} &\sum_{\sigma_a\in \mathfrak{M}(\leaf\cup\root(T^a))} \prod_{s\in\interoot(T^a)}|I_s(\sigma_a)|!  \sum_{\vec{d}_a\in\mathcal{D}(\sigma_a)} \prod_{\pi_a\in\sigma_a}  M_{\pi_a}(\vec{d}_a)\\
\times &\sum_{\sigma_b\in \mathfrak{M}(\leaf\cup\root(T^b))} \prod_{s\in\interoot(T^b)}|I_s(\sigma_b)|!  \sum_{\vec{d}_b\in\mathcal{D}(\sigma_b)} \prod_{\pi_b\in\sigma_b}  M_{\pi_b}(\vec{d}_b)
\end{split}
\een
From here on we can essentially repeat the discussion following eq.\eqref{pT1prod}: We distinguish matchings in $\sigma_a$ and $\sigma_b$ containing the vertex $u$, and those that do not. For the former, we obtain products of the form $M_{(v,i),(u,k)}(\vec{d}_a) M_{(u,k),(w,j)}(\vec{d}_b)$, which can be simplified using eq.\eqref{mergeweight}:
\ben\label{mergeweight2}
\sum_{|\alpha_{u,k}|=d_{\pi}^u}M_{(v,i),(u,k)}(\vec{d}_a)\cdot M_{(u,k),(w,j)}(\vec{d}_b) = M_{(v,i)(w,j)}(\vec{d})
\een
Expressing the matchings $\sigma_a,\sigma_b$ in terms of matchings $\sigma\in\mathfrak{M}(\leaf\cup\root(T))$ by merging the $u$-vertices as before (see fig.\ref{fig:merge} and the corresponding discussion), we pick up a factor $|I_u(\sigma)|!$ and thereby arrive at the representation \eqref{PT0} as claimed. 
\end{proof}

\subsubsection{Proof of the bound \eqref{hypothesis1} for $r=0$:}\label{subsec:startA}   Lemma \ref{lemHf} provides a compact expression for the objects of interest in our proof of theorem \ref{mainthm}. Next we would like to derive an upper bound for the r.h.s. of eq.\eqref{PT0}. This is achieved with the help of the following lemma:
\begin{lemma}\label{lembound}
Let $M(\vec{d})$ be the matrix defined in eq.\eqref{Mweight}, let $b\in\bran(T)$ be the branch of $T$ fixed in theorem \ref{mainthm}, let $\pi=[(v,i)(w,j)]\in\sigma$ and define the shorthand 
\ben
\chi_u:=\xi_u\times \begin{cases}
(1+\varepsilon) & \text{ for } u \in b(T) \\
1/\varepsilon^2 & \text{ for } u \notin b(T) \, .
\end{cases}
\een
where $\varepsilon\in (0,2^{-\dext-3})$ with $\dext$ as defined in theorem \ref{mainthm}. For any $\delta\geq 0$ one has the bounds 
\ben\label{Mijbound}
|M_{\pi}(\vec{d})|\leq \begin{cases}\frac{ (|\alpha_{v,i}|+|\alpha_{w,j}|+\delta)! \,    \prod\limits_{u\in P_{\pi}} \theta(d_{\pi}^u-d_{\pi}^{\c(u)\cap P_{\pi}}) \chi_u^{d_{\pi}^u+1}   }{ (\varepsilon^2\min\limits_{u\in \n(v)} |x_u-x_{v}|)^{1+|\alpha_{v,i}|} (\varepsilon^2\min\limits_{u\in \n(w)} |x_u-x_{w}|)^{1+|\alpha_{w,j}|} \ m^{\delta} (\varepsilon^2|x_{e}-x_{f}|)^{\delta} } 
   & \text{for }v,w\neq \root \\
 \frac{\max_{u\in \c(\root)}|x_u-x_\root|^{|\alpha_{w,j}|+1}}{\min_{u\in \n(v)}|x_u-x_v|^{|\alpha_{v,i}|+1}}  \prod_{u\in P_{\pi}} \theta(d_{\pi}^u-d_{\pi}^{\c(u)\cap P_{\pi}})\cdot   \xi_u^{d_{\pi}^u} & \text{for } w=\root
\end{cases}
\een
where we use the shorthand $P_{\pi}:=(\an(v)\setminus\an(w))\cup(\an(w)\setminus\an(v))$ %and $\nN:=\sum_{v\in\leaf\cup\root}n_v$ with $n_v$ as defined in \eqref{compopdef}, 
and where 
$e$ is the vertex closest to the root in the set $\an(v)\setminus\an(w)$ (ancestors of $v$ which are not an ancestor of $w$), and similarly for $f$ with the roles of $v$ and $w$ exchanged. 
\end{lemma}
\noindent
The straightforward but tedious proof of this lemma can be found in appendix \ref{lemm2app}. 
Using lemma \ref{lemHf} we can bound the l.h.s. of \eqref{hypothesis1} for $r=0$.  
\ben\label{freeboundhelp}
\Big| \prod_{u\in\interoot(T)} \sum_{[A_u]=D_u}(\Pro_0)(T)  \Big| \leq   \sum_{\sigma\in \mathfrak{M}(\leaf\cup\root)} \prod_{u\in\interoot}|I_u(\sigma)|!  \sum_{\vec{d}\in\mathcal{D}(\sigma)} \prod_{\pi\in\sigma} \, | M_{\pi}(\vec{d})   |  
\een
We would like to bound the product of matrix entries on the r.h.s. of this inequality with the help of lemma \ref{lembound}. For this purpose, we first note that the product of combinatoric factors can be simplified using $\sum_{\pi\in I_{u}(\sigma)}(d_{\pi}^u+1)=D_u$ and using
\ben
\prod_{\pi\in \sigma} (|\alpha_{v,i}|+|\alpha_{w,j}|+\delta)!
 \leq    (\sum_{v\in\leaf}[A_v])! \, . 
 \een
A rather non-trivial point concerns the factors $m^{\delta} |x_{e}-x_{f}|^{\delta} $ in the bound \eqref{Mijbound}. How many of these factors do we obtain in the product over $\pi\in\sigma$? 
Note that, on account of the $\theta$-functions, our bound \eqref{Mijbound} on the matrix elements $M_{\pi}$ vanishes if $\vec{d}\in\mathcal{D}(\sigma)$ contains two elements $d_{\pi}^{e}<d_{\pi}^{f}$ such that $e$ is closer to the root of $T$ than $f$. Pick a vertex $u\in \inter(T)$. If we have $\Delta(u)>0$ at that vertex, then there has to be at least one pair $[(v,i)(w,j)]=\pi\in \sigma$ such that $v,w\in \de(u)$ for the product of matrix elements not to vanish, since otherwise we would have a pair $d_{\pi}^{u}<d_{\pi}^{\c(u)}$. From lemma \ref{lembound} we know that in this case, since clearly $v,w\neq\root$, we have the freedom to generate an additional power of $1/(m\cdot \min_{e,e'\in \c(u)}|x_{e}-x_{e'}| )^{\delta_u}$. Repeating this argument at every internal vertex of $T$, we  arrive at the bound
\ben\label{prefinalbound}
\prod_{\pi\in\sigma} \, | M_{\pi}(\vec{d})   |  \leq \frac{\max\limits_{u\in \c(\root)}|x_u-x_\root|^{[A_{\root}]}      \, (\sum\limits_{v\in\leaf}[A_v])!  \, \varepsilon^{-2\sum_{v\in\vert\setminus b(T)} [A_v]  }  (1+\varepsilon)^{\sum_{w\in b(T)} [A_w] } } { \prod\limits_{v\in\leaf}  \min\limits_{u\in\n(v)}|x_u-x_v|^{[A_v]}   \prod\limits_{u\in \inter} (m\cdot\min\limits_{i,j\in \c(u)}|x_i-x_j|)^{\theta(\Delta_u)\cdot \delta_u} }  \prod_{i\in\interoot} \xi^{D_i}_i \, .
\een
Substituting the bound \eqref{prefinalbound} into \eqref{freeboundhelp} and using also the estimate 
\ben
\prod_{u\in\interoot}|I_u(\sigma)|!   \leq \nN!^{{|\interoot|}}
\een
where  $\nN:=\sum_{v\in\leaf\cup\root}n_v\leq \dext$ with $n_v$ as defined in \eqref{compopdef},
as well as
\ben
| \mathfrak{M}(\leaf\cup\root)| = (\nN-1)!! \leq \nN!
\een
and
\ben
|\mathcal{D}(\sigma)| \leq \prod_{u\in\interoot} (D_u+1)^{\nN} \leq   2^{\nN \sum_{v\in\interoot\setminus b}D_v} \prod_{u\in b(T)} (D_u+1)^{\nN}  \leq   \varepsilon^{-\sum_{v\in\vert\setminus b}[A_v]} \prod_{u\in b(T)} (D_u+1)^{\dext}  \, ,
\een
to bound the summations over $\sigma\in \mathfrak{M}(\leaf\cup\root)$ and over $\vec{d}\in\mathcal{D}(\sigma)$ in \eqref{freeboundhelp}, we finally arrive at a bound for the quantities of interest:
\ben\label{freefinalbound}
\begin{split}
\Big| \prod_{u\in\interoot(T)} \sum_{[A_u]=D_u}(\Pro_0)(T)  \Big|  &\leq  \frac{\max\limits_{u\in \c(\root)}|x_u-x_\root|^{[A_{\root}]}   } { \prod\limits_{v\in\leaf}  \min\limits_{u\in\n(v)}|x_u-x_v|^{[A_v]}   \prod\limits_{u\in \inter} (m\cdot\min\limits_{i,j
\in \c(u)}|x_i-x_j|)^{\theta(\Delta_u)\cdot \delta_u} }  \prod_{u\in\interoot} \xi^{D_u}_u   \\
&\times \nN!^{|\interoot|+1}\cdot (\sum_{v\in\leaf}[A_v])!\, \varepsilon^{-3\sum_{v\in\vert\setminus b(T)} [A_v]  }   \, \prod_{w\in b}(D_w+1)^{\dext} (1+\varepsilon)^{ [A_w] }   
\end{split}
\een
This inequality indeed implies the bound \eqref{hypothesis1} for the free field ($r=0$) if we choose the constant $K$ such that $K\geq  \nN!^{|\interoot|+1}\cdot (\sum_{v\in\leaf}[A_v])! $. 

\subsubsection{Proof of the convergence relation \eqref{hyp2} for $r=0$:}\label{subsec:startB}  

To complete the induction start, it remains to be shown that the convergence property \eqref{hyp2} holds for the free theory, i.e. we need to show that (suppressing for the moment the dependence on the coordinates $x_i$)
\ben\label{hyp2recall}
\lim_{D\to\infty}(R_{0}^{D})_{A_{1}\ldots A_{M}; A_{M+1}\ldots A_{N}}^{B}=  (\C_{0})_{A_{1}\ldots  A_{N}}^{B} - \sum_{[C]=0}^\infty  (\C_{0})_{A_{1}\ldots  A_{M}}^{C}(\C_{0})_{C A_{M+1}\ldots  A_{N}}^{B}= 0
\een
on the domain $\xi< 1$ defined by \eqref{xi1def}. In terms of our tree notation, we can write the associativity remainder as
\ben
(R_{0}^{D})_{A_{1}\ldots A_{M}; A_{M+1}\ldots A_{N}}^{B} = (\Pro_0)(T_0)-\sum_{d\leq D} \sum_{[C]= d}(\Pro_0)(T_1) \, ,
\een
where $T_0\in \mathcal{T}((A_1,\ldots,A_N,B); (x_1,\ldots, x_N))$ and  $T_1\in \mathcal{T}((A_1,\ldots,A_N,B,C); (x_1,\ldots, x_N))$ are the trees shown in figure~\ref{fig:scalingtree}. Using the bound \eqref{freefinalbound} for the r.h.s. of this equation, one can verify that the sum over $d$ is absolutely convergent on the domain $\xi<1$ in the limit $D\to\infty$ [see the discussion following eq.\eqref{thmtocor}].

Thus, it remains to show that the limit in eq.\eqref{hyp2recall} is indeed zero\footnote{This fact has been shown previously, in~\cite{Olbermann:2012uf} for the case $r=0, N=3$.}. To see this, we recall equation \eqref{Mspecial}, which we can write in the limit $D\to\infty$ and for $\xi < 1$ as (using the Leibniz rule to pull Taylor expansions out of the product)
\ben\label{assoc0prf}
\sum_{[C]=0}^\infty (\C_{0})_{A_{1}\ldots A_{M}}^{C}\,  (\C_{0})_{C A_{M+1}\ldots A_{N}}^{B} = \hspace{-.4cm}  \sum_{\sigma\in \mathfrak{M}(\leaf\cup\root (T_1))}  \prod_{\pi\in\sigma\setminus I(\sigma)}\Sigma_{\pi}  \sum_{d=0}^\infty\T_{(x_1,\ldots,x_M)\to (x_{M},\ldots, x_M)}^{d} \prod_{\pi'\in I(\sigma)}     \Sigma_{\pi'}
\een
Here $\T_{(x_1,\ldots,x_M)\to (x_{M},\ldots, x_M)}^{d}$ is the multivariate Taylor operator,
\ben
\T_{(x_1,\ldots,x_M)\to (x_{M},\ldots, x_M)}^{d} f(x_1,\ldots, x_M)= \hspace{-.3cm} \sum_{|v_1|+\ldots+|v_M|=d} \prod_{i=1}^M\frac{(x_i-x_M)^{v_i}}{v_i!}\, \partial_{y_i}^{v_i} f(y_1,\ldots,y_M)\Big|_{y_i\to x_M}\, .
\een
Using the fact that the Taylor series is convergent on the mentioned domain and recalling our explicit formula \eqref{OPEpf} for the zeroth order OPE coefficients, we therefore arrive at  the relation
\ben
\sum_{[C]=0}^\infty (\C_{0})_{A_{1}\ldots A_{M}}^{C}\,  (\C_{0})_{C A_{M+1}\ldots A_{N}}^{B} = \sum_{\sigma\in \mathfrak{M}(\vert (T_0))}  \prod_{\pi\in\sigma}\Sigma_{\pi}  =(\C_{0})_{A_{1}\ldots A_{N}}^{B}\, ,
\een
which establishes equation \eqref{hyp2} for the free field and thereby concludes the induction start.

\subsection{Induction step: Higher perturbation orders}\label{sec:step}

Assuming that theorem \ref{mainthm} holds up to perturbation order $r$, we now want to show that it also holds at order $r+1$. Our main tool to achieve this task is the recursion formula for the OPE coefficients, eq.\eqref{recursionintro}, which in turn implies a corresponding recursion formula for the expressions $(\Pro_{r})(T)$. 

\subsubsection{Proof of the bound \eqref{hypothesis1} at order $r+1$:}\label{subsec:stepA} 

When expanded in $g$, our recursion formula\footnote{Our choice of ``renormalisation scheme" enters at this stage: The particular form of the recursion formula given here was derived for the so called  BPHZ renormalisation conditions. See section \ref{sec:massless} for a discussion of renormalisation ambiguities.} \eqref{recursionintro} reads at order $g^{r+1}$:
\ben\label{indsteprecur}
\begin{split}
(\C_{r+1})_{A_1\ldots A_N}^B(x_1,\ldots, x_N) &= \frac{-1}{(r+1)} \int\d^4 y\, \Big[  (\C_{r})_{\INT A_1\ldots A_N}^B(y,x_1,\ldots, x_N)  \\
&-\sum_{i=1}^N \sum_{[C]\leq [A_i]} \sum_{r_1+r_2=r}   (\C_{r_1})_{\INT A_i}^C(y,x_i)\,   (\C_{r_2})_{A_1\ldots \widehat{A_i} C\ldots  A_N}^B(x_1,\ldots, x_N)\\
&- \sum_{[C]< [B]} \sum_{r_1+r_2=r}     (\C_{r_1})_{A_1\ldots A_N}^C(x_1,\ldots, x_N) \, (\C_{r_2})_{\INT C}^B(y,x_N)   \Big] \, ,
\end{split}
\een
where the index $\INT$ corresponds to the \emph{interaction operator} of our model, i.e. $\O_{\INT}:=\varphi^4/4!$. Formula \eqref{indsteprecur} allows us to write the l.h.s. of   \eqref{hypothesis1} at order $r+1$  in terms of $r$-th order quantities via
\ben\label{remr+1}
\begin{split}
&\prod_{i\in\interoot} \sum_{[A_{i}]=D_{i}}(\Pro_{r+1})(T)=\frac{-1}{r+1}\prod_{i\in\interoot} \sum_{[A_{i}]=D_{i}}\sum_{\sum_{u\in\inter} r_u=r+1}\prod_{v\in\inter}   (\C_{r_v})_{(A_w)_{w\in \c(v)}}^{A_v}((x_i)_{i\in \c(v)};x_v) \\
&=\frac{-1}{r+1}\prod_{i\in\interoot} \sum_{[A_{i}]=D_{i}} \sum_{v\in \inter} \int_{y} \Big[ (\Pro_{r})(T_v)-\sum_{w\in \c(v)} \sum_{[A_u]\leq [A_w]} (\Pro_{r})(T_{w,A_u})  -  \sum_{[A_u]<[A_{v}]} (\Pro_{r})(T_{A_u,v})  \Big]
\end{split}
\een
where the trees $T_v,T_{v,A_u},T_{A_u,v}$ are obtained form $T\in\mathcal{T}(\vec{x};\vec{A})$ as follows (see fig.\ref{fig:perttree}):
\begin{itemize}
\item $T_v\in\mathcal{T}((\vec{x},y); (\vec{A},\INT))$ is obtained from $T$ by connecting an additional leaf with weight $(\INT, y)$ to the vertex $v$.
\item $T_{v,A_u}\in\mathcal{T}((\vec{x},y); (\vec{A},\INT,A_u))$ is obtained by connecting a leaf  with weight $(\INT, y)$ to the parent edge of $v$, splitting this edge into two halves. The new vertex $u$ adjacent to these two halves receives the weight $(A_u,x_v)$.
\item $T_{A_u,v}\in\mathcal{T}((\vec{x},y); (\vec{A},\INT,A_u))$ is obtained by connecting a leaf  with weight $(\INT, y)$ to the parent edge of $v$, splitting this edge into two halves (if $v=\root$, then we add a parent edge to $v$ and connect the leaf to this new root). The new vertex $u$ adjacent to these two halves receives the weight $(A_v,x_v)$, and we change the weight of the vertex $v$ to $(A_u,x_v)$.
\end{itemize}
\begin{figure}[htbp]
\begin{center}
\includegraphics[width=\textwidth]{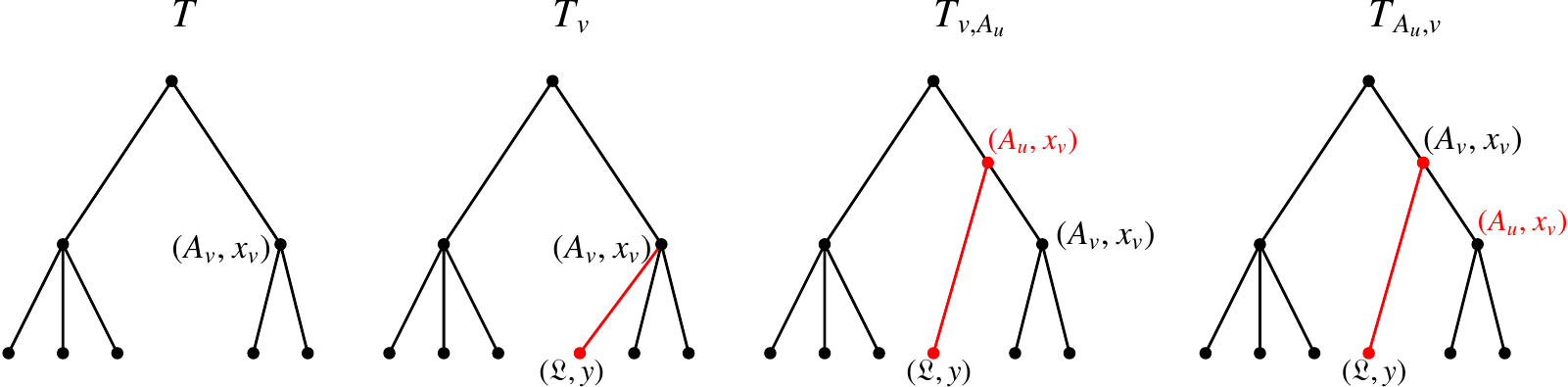}
\end{center}
\caption{The trees  $T_v,T_{v,A_u},T_{A_u,v}$ are obtained from the tree $T$ by adding an external edge.}
\label{fig:perttree}
\end{figure}

\noindent Our plan is now to combine the formula \eqref{remr+1} with the inductive bound \eqref{hypothesis1}, which holds up to order $r$ by assumption, in order to verify the bound \eqref{hypothesis1} at order $r+1$. The terms under the integral in eq.\eqref{remr+1} can be estimated with the help of the following bounds:

\begin{lemma}\label{lem:amp}
Denote by $B_r(T)$ the r.h.s. of  \eqref{hypothesis1}. Then 
\begin{align}
&\Big|\prod_{i\in\interoot} \sum_{[A_{i}]=D_{i}} (\Pro_r)(T_{v}) \Big| \leq   \frac{B_r(T)\, K  }{\min\limits_{w\in \c(v)} |x_w-y|^4}   \sup\left(\frac{\min\limits_{i,j\in \c(v)}|x_j-x_i| }{\min\limits_{i\in \c(v)}|y-x_i|}, \frac{|y-x_{v}|}{\max\limits_{i\in \c(v)}|x_i-x_{v}|} ,  1\right)^{2\delta_v}   \non\\
&\times  \left(\frac{\prod_{w\in b}(D_{w}+1)}{\varepsilon}\right)^{4\cdot\ex^{r+1}} \hspace{-.4cm}  \prod_{i\in \c(v)} \sup\left(\frac{\min\limits_{j\in \c(v)\setminus\{i\}}|x_j-x_i| }{|y-x_i|}, 1\right)^{[A_i]}\hspace{-.2cm} \sup\left(\frac{|y-x_v|}{\max\limits_{j\in \c(v)\setminus\{i\}}|x_j-x_v| }, 1\right)^{[A_v]}  
\end{align}
\begin{align}
\Big| \sum_{[A_{u}]=d}\prod_{i\in\interoot} \sum_{[A_{i}]=D_{i}}(\Pro_r)(T_{v,A_{u}})\Big| &\leq    \frac{ B_r(T)\, K  \sup(1,1/m|x_v-y|)^{\delta_{u}} }{m^{\theta([A_v]+4-d)\delta_{u}}|x_v-y|^{4+\theta([A_v]+4-d)\delta_{u}}} \left(\frac{\min\limits_{i\in \n(v)}|x_v-x_i| } {|y-x_v|} \right)^{[A_v]-d}      \non \\ % \log_+^r\left( \frac{\max\limits_{i\in c(t(e))}|x_i-x_e|}{\varepsilon^{r\cdot 2^r}|x_e-x_0|} \right)
&\times    \sup\left(\frac{\min\limits_{i \in \n(v)}|x_i-x_v| }{|y-x_v|} ,1  \right)^{\delta_{\p(v)}}  \hspace{-.3cm}\left(\frac{\prod_{w\in b}(D_{w}+1)}{\varepsilon}\right)^{4\cdot\ex^{r+1}}    \hspace{-.3cm}  \chi(v,d)  
   \\
\Big| \sum_{[A_{u}]=d}\prod_{i\in\interoot} \sum_{[A_{i}]=D_{i}}(\Pro_r)(T_{A_{u},v})\Big| &\leq    \frac{ B_r(T)\, K  \sup(1,1/m|x_v-y|)^{\delta_{u}}  }{m^{\theta(d+4-[A_v])\delta_{u}}|x_v-y|^{4+\theta(d+4-[A_v])\delta_{u}}} \left(\frac{|y-x_v|} {\max\limits_{i\in \c(v)}|x_v-x_i| } \right)^{[A_v]-d}   \non \\ %\cdot \log_+^r\left( \frac{\max\limits_{i\in c(t(e))}|x_i-x_e|}{\varepsilon^{r\cdot 2^r}|x_e-x_0|} \right)
& \times     \sup\left(\frac{\min\limits_{i\neq j \in \c(v)}|x_i-x_j| }{|y-x_v|} ,1 \right)^{\delta_{v}}   \left(\frac{\prod_{w\in b}(D_{w}+1)}{\varepsilon}\right)^{4\cdot\ex^{r+1}}    \hspace{-.3cm}  \chi(v,d)
\end{align}
where $K>0$ is a constant that depends neither on the integers $D_i$ nor on $\varepsilon$ or the $\delta_v$, and where we use the shorthand
\ben
\chi(v,d)= \begin{cases}
  (1+\varepsilon)^{\ex^{r+1}\, d} \cdot (d+1)^{\ex^{r+1}(\dext+4)} & \text{if }v\in b \\
 \varepsilon^{-\ex^{r+1} d} & \text{if }v\notin b \, .
\end{cases}
\een
\end{lemma}
\begin{proof}
The lemma follows by straightforward computation from the inductive bound~\eqref{hypothesis1}. 
\end{proof}

\noindent We now substitute these bounds under the integral in the recursion formula \eqref{remr+1}. It is, however, not possible to estimate the resulting individual terms directly, because the integral over each individual term contains divergences in the regions where $y\approx x_i$ (UV) or where $|y|\gg \sup_i|x_i|$ (IR). As mentioned in our overview of the proof at the beginning of this section, we therefore have to take a little more care and take into account cancellations between these divergent terms for each of those dangerous regions.  In order to study these cancellations of singularities at short- and large distances, we define the following partition of $\mathbb{R}^4$:
\begin{defn}[Integration regions]
Let $v\in \inter(T)$ be an internal vertex of the tree $T$ and let $b\in\bran(T)$ be the branch mentioned in theorem \ref{mainthm}. Then 
\begin{center}\textbf{(UV-regions)}\end{center}
\ben
 \Omega^v_i:=\begin{cases}
\left\{y\in\mathbb{R}^4\, :\,  |x_i- y | \cdot (1+\varepsilon)^{2\cdot \ex^{r+1}} < \min_{j\in \n(i)} |x_i-x_j| \right\}  & \text{ if }   i\in \c(v)\cap b(T)   \\ 
 & \\
\left\{y\in\mathbb{R}^4\, :\,  |x_i- y | \cdot \varepsilon^{-2\cdot \ex^{r+1}} < \min_{j\in \n(i)} |x_i-x_j| \right\}  & \text{ if }   i\in \c(v) \setminus b(T)  \\
 \end{cases}
\een
\begin{center}\textbf{(IR-region)}\end{center}
\ben
 \Omega^v_{IR}: =  \begin{cases}
 \{y\in\mathbb{R}^4\, :\,  |x_{v}- y |  \geq \max\limits_{j \in \c(v)} |x_{v}-x_j|  \cdot (1+\varepsilon)^{2\cdot \ex^{r+1}}  \}\setminus \cup_i \Omega^{v}_i  & \text{if }\p(v)\in b(T) \\
  & \\
 \{y\in\mathbb{R}^4\, :\,  |x_{v}- y |  \geq \max\limits_{j \in \c(v)} |x_{v}-x_j|  \cdot \varepsilon^{-2\cdot \ex^{r+1}}  \}\setminus \cup_i \Omega^{v}_i  & \text{if }\p(v)\notin b(T) \\
  \end{cases}
 \een
 \begin{center}\textbf{(Intermediate region)}\end{center}
 \ben
\Omega^v_{IM}:=\mathbb{R}^4\setminus ( \cup_i \Omega^v_i\cup \Omega^v_{IR})
\een
\end{defn}
\paragraph{Remark:} Note that for any $v\in\inter(T)$ one has $\Omega_{IM}^v\cup \Omega_{IR}^v\cup_{i\in \c(v)} \Omega_{i}^v=\mathbb{R}^4$ and that these sets are disjoint, in particular $\Omega^v_i \cap \Omega^v_j=\emptyset$   if $i\neq j$. Note further that the UV- and IR-regions get smaller as we increase the perturbation order, which will be needed later in order to obtain sufficiently strong bounds within those regions [more precisely, this fact is going to be crucial for the estimates \eqref{uvest} and \eqref{PIRbound}].

\vspace{.5cm}

\noindent We now derive a bound on the r.h.s. of  \eqref{remr+1} by decomposing the $y$-integral into integrals over the regions defined above. We will see that, indeed, the resulting bounds for the contributions from each of those regions are consistent with \eqref{hypothesis1} at order $r+1$.

\paragraph{The intermediate distance region $ \Omega^v_{IM}$:} In this region the integration variable $y$ of eq.\eqref{remr+1} is neither very close to, nor very far from the points $(x_i)_{i\in \c(v)}$. Hence, we will encounter neither UV- nor IR-divergences, and we can simply insert the bounds from lemma \ref{lem:amp} in order to estimate the integrand, without taking into account any further cancellations. 

Let us fix an internal vertex $v\in \inter(T)$. By definition, we then have for any $e\in \c(v)$%and any vertex $v$, 
\ben\label{IMratio}
\frac{\min_{i \in \c(v)}|x_e-x_i| }{|y-x_e|}   \leq \begin{cases}
(1+\varepsilon)^{2\cdot \ex^{r+1}} & \text{ for } e \in b(T) \\
\varepsilon^{-2\cdot \ex^{r+1}} & \text{ for } e \notin b(T) 
\end{cases}  
\Bigg\}
\geq \frac{|y-x_e|}{\max_{i \in \c(v)}|x_e-x_i| } \, .
\een
Furthermore, we have the inequality 
\ben\label{intest1}
\begin{split}
 \int\limits_{\Omega^v_{IM}} \frac{\d^4 y}{\min\limits_{i\in \c(v)} |y-x_i|^4} &\leq  \frac{(2\pi)^2}{\varepsilon^{2\delta \cdot \ex^{r+1}}}  \sum_{i\in \c(v)}  \int\limits_{0}^{\frac{|x_i-x_v|}{\varepsilon^{2\cdot \ex^{r+1}}}}
 %\int\limits_{|x_0-x_i|=\varepsilon^{2\cdot 4^{r+1}} \min |x_i-x_j| }^{|x_0-x_i|=(1+\varepsilon)^{2\cdot 4^{r+1}}\max|x_i-x_{p(v)}|} 
 \frac{  \, \d |y|}{|y|^{1-\delta} \min\limits_{j\in\n(i)}|x_i-x_j|^\delta }\\
 &  \leq (2\pi)^2 \, \nN\,   \left(   \frac{\max\limits_{i\in\c(v)}|x_i-x_{v}|}{\varepsilon^{4\cdot \ex^{r+1}}\!\!\min\limits_{i\neq j\in \c(v)}|x_i-x_j|}\right)^{\delta}
\end{split}
\een
where as before $\nN:=\sum_{v\in\leaf\cup\root}n_v$ is the ``total number of fields" associated to the external vertices of the tree $T$. Combining these inequalities with lemma \ref{lem:amp} and choosing $\delta$ sufficiently small such that $\delta+\delta_v<1$, we obtain for the first term under the integral in \eqref{remr+1}  the bound 
\ben
\begin{split}\label{IMabound}
 &\int\limits_{ \Omega^v_{IM}}  \Big|\prod_{i\in\interoot} \sum_{[A_{i}]=D_{i}} (\Pro_r)(T_{v}) \Big|\,   \d^4 y \\
 &\leq   B_r(T) \,K \, \left(\frac{ \prod_{w\in b}(D_{w}+1)}{\varepsilon^{1+\delta} }\right)^{4 \cdot\ex^{r+1}} \left(\frac{(1+\varepsilon)^{\sum_{e\in (\c(v)\cup v)\cap b} [A_e]}}{ \varepsilon^{\sum_{e\in (\c(v)\cup v)\setminus b} [A_e]} }\right)^{2\cdot \ex^{r+1}} \left( \frac{\max\limits_{i\in \c(v)}|x_i-x_{v}|}{\min\limits_{i\neq j \in \c(v)}|x_i-x_j|} \right)^{\delta}
\end{split}
\een
where constants (i.e. factors depending neither on the weights $D_i$ nor on $\varepsilon$) were absorbed into $K$. The last factor on the r.h.s. can be absorbed into the expression $B_r(T)$ by adjusting the parameter $\delta_v\to\delta_v+\delta\in(0,1)$.  To see that the resulting bound is smaller than the r.h.s. of \eqref{hypothesis1} at order $r+1$, we note that the inductive bound  \eqref{hypothesis1} grows like
\ben\label{boundinc}
\begin{split}
B_{r+1}(T)&= B_r(T)\ K\  \left(\frac{ (1+\varepsilon)^{\sum_{w\in b} [A_w] }   }{ \varepsilon^{\sum_{v\in\vert\setminus b}[A_{v}]} } \cdot \prod_{w\in b}(D_{w}+1)^{\dext}\right)^{7\cdot \ex^{r+1} }  
\end{split}
\een
as we increase the perturbation order $r$, where $K$ is some constant that depends neither on the $D_{i}$ nor on $\varepsilon$ of the $\delta_{v}$.  Since the remaining terms on the r.h.s. of \eqref{IMabound} are indeed smaller than the r.h.s. of~\eqref{boundinc} (choosing $\delta<1/4$ and assuming that $\sum_{v\in\vert\setminus b}[A_{v}]>0$), we conclude that this contribution to the recursion formula~\eqref{remr+1} is consistent with the claimed bound~\eqref{hypothesis1}.

 Similarly, using lemma \ref{lem:amp} as well as the estimates \eqref{IMratio} and \eqref{intest1}, we obtain for any $w\in \c(v)$ the following bound on the second term under the integral in \eqref{remr+1}:
\ben\label{IMebound}
\begin{split}
&\int\limits_{\Omega^v_{IM}} \Big| \sum_{[A_{u}]\leq D_w}\prod_{i\in\interoot} \sum_{[A_{i}]=D_{i}}(\Pro_r)(T_{w,A_{u}})\Big|\, \d^4 y  \leq B_r(T)\, K\, \left(\frac{\prod_{i\in b}(D_{i}+1)}{\varepsilon^{2\delta_u+1}} \right)^{4\cdot \ex^{r+1}} \\
&  \quad \times \left( \frac{\max\limits_{i\in \c(v)}(1/m,|x_i-x_{v}|)}{\min\limits_{i\neq j \in \c(v)}|x_i-x_j|} \right)^{3\delta_u}   \times
 \begin{cases} (D_w+1)^{\ex^{r+1} (\dext+4)+1}
  (1+\varepsilon)^{2\cdot \ex^{r+1}\, (D_w+\delta_v)}  & \text{if }w\in b \\
 (D_w+1) \, \varepsilon^{-2(D_{w}+\delta_{v})\ex^{r+1}}   & \text{if }w\notin b 
\end{cases} 
\end{split}
\een
The factor with exponent $\delta_u$ can again be absorbed into $B_r(T)$ by choosing $\delta_u$ sufficiently small and increasing the value of $\delta_v$ slightly. One checks, using also the inequality $(D_w+1)\leq  \varepsilon^{-D_w}$ for the case $w\notin b$, that the bound \eqref{IMebound} is indeed  smaller than \eqref{boundinc}, and it is therefore consistent with our hypothesis \eqref{hypothesis1}. For the third term on the r.h.s. of  \eqref{remr+1} we can proceed in essentially the same manner as for the second one and we find that also the integral over $| \sum_{[A_{u}]\leq D_v}\prod_{i\in\interoot} \sum_{[A_{i}]=D_{i}}(\Pro_r)(T_{A_{u}, v})|$ satisfies the bound \eqref{IMebound}. 

Thus, we have found that the contributions from each term under the integral are smaller than the claimed bound \eqref{hypothesis1}. In order to bound the total contribution from this integration region, it remains to estimate the number of terms appearing under the integral. For a given $v$, the integrand in  \eqref{remr+1} contains $|\c(v)|+2\leq \nN+2$ terms. The sum over all vertices $v$ contains $|\inter(T)|<\nN$ terms. Both of these factors can be absorbed into the constant $K$ in our bound. 

\emph{To summarise, we have verified that the contribution to the r.h.s. of recursion formula \eqref{remr+1} from the intermediate integration region is smaller than the claimed bound \eqref{hypothesis1}.}

\paragraph{The UV-regions $\Omega^v_i$:} Here the integration variable $y$ is close to one of the points $x_i$, so we have to take into account cancellations between different terms under the integral in the recursion formula \eqref{remr+1}. In order to achieve this, we not only  make use of the inductive bound \eqref{hypothesis1} here, but we also apply the induction hypothesis \eqref{hyp2} stated in theorem \ref{mainthm} in order to organise the short distance cancellations. 

Fix a $v\in \inter(T)$ and a $w\in \c(v)$ and consider now $y\in\Omega_w^v$. To bound the integral over the expressions $(\Pro_r)(T_{A_{u},v})$  and $(\Pro_r)(T_{i,A_{u}})$with $i\in\c(v)\setminus\{w\}$, we can proceed as above and arrive at the same bounds as in the intermediate region $y\in\Omega_{IM}^v$. For the two remaining terms under the integral, our second induction hypothesis, eq.\eqref{hyp2}, implies 
\ben
(\Pro_r)(T_{v})-\sum_{[A_u]\leq D_w}(\Pro_r)(T_{w,A_{u}}) =\sum_{[A_u]>D_w}(\Pro_r)(T_{w,A_{u}}) \, . 
\een
To bound the r.h.s. of this equation, we now use lemma \ref{lem:amp}, distinguishing the cases $w\in b(T)$ and $w\notin b(T)$ in the process. Making use of the inequality
\ben
\frac{|y-x_w|} {\min\limits_{i\in \n(w)}|x_w-x_i| } \leq  \begin{cases}
  (1+\varepsilon)^{-2\cdot \ex^{r+1}} & \text{ if }w\in b(T) \\
   \varepsilon^{2\cdot \ex^{r+1}} & \text{ if }w\notin b(T)
\end{cases}
\quad \text{ for } y\in\Omega^v_w\, ,
\een
we obtain the bound
\ben\label{uvest}
\begin{split}
&\int_{\Omega_w^v}  \Big| \sum_{[A_{u}]> D_w}\prod_{i\in\interoot} \sum_{[A_{i}]=D_{i}}(\Pro_r)(T_{w,A_{u}})\Big|\, \d^4 y \\
& \leq  \int_{\Omega^v_w}\d^4 y\, \sum_{d=0}^\infty   \frac{B_r(T)\ K \ \max(1,1/m|y-x_w|)^{\delta_u} }{|x_w-y|^{3+\theta(3-d)\delta_u }m^{\theta(3-d)\delta_0}  \min\limits_{i\in \n(w)}|x_w-x_i|}  \,  \left(\frac{\prod_{i\in b}([A_i]+1)}{\varepsilon}\right)^{4\cdot\ex^{r+1}} \\
&\qquad\qquad\times   \begin{cases}   
  (1+\varepsilon)^{\ex^{r+1}\, (D_w-d+1)}(d+D_w+2)^{\ex^{r+1}(\dext+4)}\, 
  & \text{ for }w\in b(T) \\
  \varepsilon^{-\ex^{r+1}\, (D_w-d+1)}  & \text{ for }w\notin b(T)
  \end{cases}
    \\
&\leq B_r(T)  \, K\,\left(\frac{\prod_{i\in b}([A_i]+1)}{\varepsilon}\right)^{4\cdot\ex^{r+1}}\,   \sup\Big(1, \frac{1}{m\cdot \min\limits_{i\in \n(w)}|x_w-x_i|}\Big)^{2\delta_u}  \\
&\qquad\qquad\times  \begin{cases}
  (1+\varepsilon)^{\ex^{r+1}\, (D_w+2)}\,  \varepsilon^{-2\cdot\ex^{r+1}(\dext+4)-2} \, (\ex^{r+1}(\dext+4))! &  \text{ for }w\in b(T) \\
 \varepsilon^{-\ex^{r+1}\, (D_w+1)-1}  &  \text{ for }w\notin b(T) 
\end{cases}
\end{split}
\een
Here we used the inequality
\ben\label{seriesbound}
%\begin{split}
\sum_{d=0}^{\infty}   (1+\varepsilon)^{-d\,\ex^{r+1}}   (D_w+d+2)^{\ex^{r+1}(\dext+4)}
%& \leq (1+\varepsilon)^{(D_w+2)\ex^{r+1}} \sum_{d=0}^{\infty}   (1+\varepsilon)^{-d\,\ex^{r+1}}   \frac{(d+\ex^{r+1}(\dext+4))!}{d!}   \\
%&= (1+\varepsilon)^{(D_w+2)\ex^{r+1}} \frac{(\ex^{r+1}(\dext+4))!}{(1-(1+\varepsilon)^{-\ex^{r+1}})^{\ex^{r+1}(\dext+4)}} \\
 \leq \frac{ (1+\varepsilon)^{(D_w+2)\ex^{r+1}}}{ \varepsilon^{2\cdot \ex^{r+1}(\dext+4)+2} } \, (\ex^{r+1}(\dext+4))!
%\end{split}
\een
as well as the elementary estimate 
\ben
\sum_{d=0}^\infty \varepsilon^{-\ex^{r+1} d} \leq \frac{1}{1-\varepsilon} \leq \frac{1}{\varepsilon}
\een
to bound the infinite sums and the inequality (choosing $\delta_u<1/2$)
\ben\label{logbd}
\int_{\Omega_w^v} \frac{\max(1,1/m|y-x_w|)^{\delta_u} \, \d^4 y }{|x_w-y|^{3+\theta(3-d)\delta_u }m^{\theta(3-d)\delta_u}  \min\limits_{i\in \n(w)}|x_w-x_i|}  \leq (2\pi)^{2} \sup\Big(1, \frac{1}{m\cdot \min\limits_{i\in \n(w)}|x_w-x_i|}\Big)^{2\delta_u}
\een
to bound the $y$-integral. Choosing $\delta_u$ small enough such that $\delta_v+2\delta_u<1$, we can absorb the factor with exponent $2\delta_u$ into the bound $B_r(T)$ via a redefinition of $\delta_v$. The factorial $(\ex^{r+1}(\dext+4))!$  can be absorbed into the constant $K$.  As the remaining terms on the r.h.s. of \eqref{uvest} are smaller than \eqref{boundinc}, we conclude that also this contribution to  \eqref{remr+1}  is consistent with our inductive bound \eqref{hypothesis1}.

\emph{To summarise, we have verified that contributions from the integral over the short distance regions $\Omega_{w}^v$ to the r.h.s. of  \eqref{remr+1} are smaller than the claimed bound \eqref{hypothesis1}.}

\paragraph{The  IR-region $\Omega^v_{IR}$:} Here the integration variable $y$ is far away form the points $x_i$, and we again have to take into account cancellations between different terms under the integral in order to bound this contribution to the recursion formula \eqref{remr+1}.

Fix a vertex $v\in \inter(T)$. For the second term on the r.h.s. of \eqref{remr+1} we can proceed essentially as in the case of $y\in \Omega^v_{IM}$ before. The only difference here is that instead of \eqref{intest1} we use the inequality 
\ben\label{intest2}
 \int\limits_{\Omega^v_{IR}} \frac{\d^4 y}{m^{\delta_u} |y-x_w|^{4+\delta_u}} \leq \frac{(2\pi)^{2}}{(m\min_{i\in\n(w)}|x_w-x_i|)^{\delta_u}}
\een
to bound the integral over $y$. This factor can be absorbed into a redefinition of $\delta_v$ as explained previously below \eqref{logbd}. 

 In order to find useful bounds on the remaining terms, we again have to make use of our second induction hypothesis, eq.\eqref{hyp2}, which implies that
\ben
(\Pro_r)(T_{v})-\sum_{[A_u]< D_v}(\Pro_r)(T_{A_{u},v}) =\sum_{[A_u]\geq D_v}(\Pro_r)(T_{A_{u},v}) 
\een
for $y\in\Omega_{IR}^v$. Lemma \ref{lem:amp} then implies for the r.h.s.
\ben\label{PIRbound}
\begin{split}
&\int_{\Omega^v_{IR}}  \Big| \sum_{[A_{u}]\geq D_v}\prod_{i\in\interoot} \sum_{[A_{i}]=D_{i}}(\Pro_r)(T_{A_{u},v})\Big| \, \d^4 y \\
& \leq  \int_{\Omega^v_{IR}}\d^4 y\, \sum_{d=D_{v}}^\infty   \frac{B_r(T)\ K \ \max(1,1/m|y-x_v|)^{\delta_u} }{|x_v-y|^{4+\delta_u }m^{\delta_u} } \   \left(\frac{\prod_{i\in b}([A_i]+1)}{\varepsilon}\right)^{4\cdot\ex^{r+1}}  \\
&\qquad \times     
\begin{cases}   
  (1+\varepsilon)^{\ex^{r+1}\, (2D_v-d)}(d+1)^{\ex^{r+1}(\dext+4)}\, 
  & \text{ for }v\in b(T) \\
  \varepsilon^{-\ex^{r+1}\, (2D_v-d)}  & \text{ for }v\notin b(T)
  \end{cases} \\
& \leq \frac{1}{(m\cdot \max_{j\in \c(v)}|x_{v}-x_j|)^{\delta_u}  }  \cdot  
\text{r.h.s. of \eqref{uvest}} 
\end{split}
\een
Here we used the same estimates as in the short-distance case to bound the sum over $d$, and we used
\ben\label{intest3}
\int_{y\in\Omega^v_{IR}} \frac{\d^4 y}{|x_v-y|^{4+\delta}} \leq  \frac{(2\pi)^{2}}{\max_{j\in \c(v)}|x_{v}-x_j|^\delta  }  \, ,
\een
to bound the $y$-integral. Choosing $\delta_u$ small enough, we can absorb the first factor on the r.h.s. of \eqref{PIRbound} into a redefinition of $\delta_v$.
The remaining terms in the bound \eqref{PIRbound} are then smaller than \eqref{boundinc}, and \emph{we conclude that also this contribution to \eqref{remr+1} is consistent with the claimed bound \eqref{hypothesis1}.}

\vspace{.5cm}

Combining our bounds for the intermediate-, UV- and IR-regions, we conclude that the r.h.s of  \eqref{remr+1} satisfies a bound that is smaller than our hypothesis \eqref{hypothesis1} at order $r+1$.

\subsubsection{Proof of the convergence relation \eqref{hyp2} at order $r+1$:}\label{subsec:stepB}

The last step in the induction is to show, assuming that theorem \ref{mainthm} holds up to perturbation order $r$, that the second statement of the theorem, eq.\eqref{hyp2}, holds also at order $r+1$. For this purpose, we write down the recursion relation for the remainder, i.e
\ben\label{recurR}
\begin{split}
&\lim_{D\to\infty} (R_{r+1}^D)_{A_1\ldots A_M ; A_{M+1}\ldots A_N }^B=  (\Pro_{r+1})(T_0)-\lim_{D\to\infty}  \sum_{d\leq D} \sum_{[C]= d}(\Pro_{r+1})(T_1) \\
&= \lim_{D\to\infty}\int_{\mathbb{R}^4}\d y^4 \Bigg\{  (\C_{r})_{ \INT A_{1}\ldots A_{N}}^{B} - \sum_{\substack{[C]\leq D \\ s\leq r } }   (\C_{s})_{\INT A_{1}\ldots A_{M}}^{C}
 (\C_{r-s})_{C A_{M+1}\ldots A_{N}}^{B} \\
&- \sum_{\substack{[C]\leq D \\ s\leq r } }  (\C_{s})_{ A_{1}\ldots A_{M}}^{C}\Bigg(   (\C_{r-s})_{\INT C  A_{M+1}\ldots A_{N}}^{B} -  \sum_{\substack{[C'] \leq D \\ t\leq r-s } }  (\C_{t})_{ \INT C}^{C'}   (\C_{r-s-t})_{C' A_{M+1}\ldots A_{N}}^{B} \Bigg)\\
&- \sum_{i=1}^N \sum_{\substack{[C]\leq [A_i] \\ s\leq r } } (\C_{s})_{\INT A_{i}}^{C}\, \Bigg( (\C_{r-s})_{ A_{1}\ldots \widehat{A_i}C\ldots A_{N}}^{B} -   \sum_{\substack{[C'] \leq D \\ t\leq r-s }}  \underbrace{(\C_{t})_{ A_1 \ldots A_M}^{C'}   (\C_{r-s-t})_{C' A_{M+1}\ldots A_{N}}^{B}}_{A_i \to C} \Bigg)\\
&- \sum_{\substack{[C] < [B] \\ s\leq r } }  \Bigg( (\C_{s})_{A_1\ldots  A_{N}}^{C}\,   - \sum_{\substack{[C'] \leq D \\ t\leq s }}   (\C_{t})_{ A_1 \ldots A_M}^{C'}   (\C_{s-t})_{C' A_{M+1}\ldots A_{N}}^{C}  \Bigg) (\C_{r-s})_{ \INT C}^{B}  \Bigg\} \\
\end{split}
\een
where $T_0\in \mathcal{T}((A_1,\ldots,A_N,B); (x_1,\ldots, x_N))$ and  $T_1\in \mathcal{T}((A_1,\ldots,A_N,B,C); (x_1,\ldots, x_N))$ are the trees depicted in fig.\ref{fig:scalingtree}.  In order to show that this expression vanishes under the assumption $\xi <1$, we would like to exchange the order of the integral and the limit. By the \emph{dominated convergence theorem}, this is allowed under the following conditions:
\begin{enumerate}
\item For all $D\in\mathbb{N}$ the integrand is bounded by some integrable function $B(y)$.
\item The limit $D\to\infty$ of the integrand converges pointwise almost everywhere. 
\end{enumerate}
The first condition is easily checked with the help of the bounds derived in the previous section combined with the inequality \eqref{dsumineq} to bound the sum over $[C]$. For the bounding function $B(y)$ we can choose for example
\ben
B(y):= %  \left(
\frac{B_{r+1}(T_0)}{(1- \xi (1+\varepsilon)^{\ex^{r+1}})^{\ex^{r+1}\dext+1}}
\cdot \min\left( \frac{\min|x_i-x_j|^{-1+\delta}}{\min |x_i-y|^{3+\delta}} ,    \frac{m^{-\delta}}{\min |x_i-y|^{4+\delta}} \right)
\een
for some $\delta\in (0,1)$ and for $\varepsilon\in(0,2^{-\dext-4r-3}]$, where $\dext=\sum_i [A_i]+[B]$. To show that the integrand converges pointwise to a limit as $D\to\infty$, we make the following observations:
Using our induction hypothesis \eqref{hyp2} at order $r$, it immediately follows that the last two lines of \eqref{recurR} vanish in the limit under the assumption $\xi < 1$. To treat the remaining terms under the integral, we have to take a little more care: Consider first the region 
\ben
\Omega_1:=\{y\in\mathbb{R}^4: |x_M-y| (1+\varepsilon)^{2\cdot \ex^{r+1}} < \min_{M<j\leq N}|x_M-x_j|\, , \, |y-x_i|>0 \}  
\een
for some small $\varepsilon>0$. In that case, the first two terms under the integral in \eqref{recurR} cancel in the limit $D\to\infty$ by our hypothesis \eqref{hyp2}, and the remaining terms under the integral are of the form 
\ben
\begin{split}
&\lim_{D\to\infty}  \Big| \sum_{\substack{[C]\leq D \\ s\leq r } }  (\C_{s})_{ A_{1}\ldots A_{M}}^{C}\Bigg(   (\C_{r-s})_{\INT C  A_{M+1}\ldots A_{N}}^{B} -  \sum_{\substack{[C'] \leq D \\ t\leq r-s } }  (\C_{t})_{ \INT C}^{C'}   (\C_{r-s-t})_{C' A_{M+1}\ldots A_{N}}^{B} \Bigg)\,  \Big|\\
&=\lim_{D\to\infty} \big|  \sum_{[C]\leq D} \sum_{[A_w]>D}  (\Pro_r)((T_1)_{u\in\interoot,A_w}) \big|  \leq   \frac{\min|x_i-x_j|^{-1+\delta}}{\min |x_i-y|^{3+\delta}}   \lim_{[C]\to\infty} B_{r+1}(T_1) = 0\, .
\end{split}
\een
The first equality follows simply from eq.\eqref{hyp2} at order $r$, and the estimate in the third line follows analogously to our discussion of the short distance region  in section \ref{subsec:stepA} [see \eqref{uvest}]. Thus, we find that for $y\in\Omega_1$ the integrand converges to 0 as $D\to\infty$. 

 In the  region 
 \ben
 \Omega_2:=\{y\in\mathbb{R}^4 : |x_M-y|\geq (1+\varepsilon)^{2\cdot \ex^{r+1}}  \max_{1\leq j\leq M}|x_M-x_j| \, , \, |y-x_i|>0 \}
 \een
we simply exchange the role of the second and third term on the r.h.s. of \eqref{recurR} and otherwise proceed in a similar manner, using estimates from the previous discussion of the large distance region $\Omega_{IR}$ [see \eqref{PIRbound}]. We find that the integrand also vanishes in this region. Note, using the assumption $\xi<1$ and choosing $\varepsilon$ sufficiently small, that the two  regions $\Omega_1$ and $\Omega_2$ cover all of $\mathbb{R}^4$ apart from the zero measure set $\{y=x_i, i\leq N\}$. Thus, we conclude that the integrand converges pointwise to $0$  almost everywhere.

To summarise, we have verified that we are allowed to exchange the order of the integral and the limit in \eqref{recurR}. Since the integrand vanishes in the limit, the same is true for the integral, which  establishes the second statement of theorem \ref{mainthm} at order $r+1$, thereby closing the induction and finishing the proof of theorem \ref{mainthm}. \hfill \qedsymbol

\section{Massless fields}\label{sec:massless}

The associativity proof for the OPE presented in section \ref{sec:proof} was restricted to the case of massive fields, $m^2>0$. In fact, the main ingredient in our construction, the 
recursion formula \eqref{recursionintro}, only holds for massive fields as stated. In the naive massless limit, the right side of the recursion formula becomes ill defined already at 
first order in $g$. This feature, however, does not indicate a fundamental problem with our approach, but is basically due to the fact that our definition of the composite operators (implicit
in our recursion formula) is unsuitable for $m^2=0$. To get around this, we will first apply a field redefinition (for $m^2 > 0$) as introduced in definition \ref{fieldred} of section \ref{sec:AX}. 
A field redefinition changes the OPE coefficients as in eq.\eqref{redefOPEax}. Consequently, these will also satisfy an appropriately modified version of the recursion formula \eqref{recursionintro}. It turns out that 
a field redefinition (depending on an arbitrary ``scale'' $L >0$) can be found such that the modified recursion relations possess a well-defined limit $m^{2}\to 0$. At this stage, the same procedure as in 
the massive case can then be applied to prove the associativity property claimed in theorem \ref{thmassoc} also for massless fields. 

\subsection{Recursion formula for massless fields}

Consider  a field redefinition in the sense of definition \ref{fieldred}, which is written in terms of a  mixing matrix $Z_A^B\in\mathbb{C}\llbracket {g}\rrbracket$ as
\ben\label{Zhat}
\widehat{\O}_A=\sum_B \, Z_A^B \O_B\, ,
\een
where $\widehat{\O}_A$ are the \emph{redefined} fields. The matrix $Z_A^B$ has to be invertible in the sense of formal power series and it has to be ``upper triangular'' in the sense that  $Z_A^B=0$ for all $[B]>[A]$. The corresponding transformation for the OPE coefficients is given by [compare~\eqref{redefOPEax}]
\ben\label{tildetrafo}
\widehat\C_{A_1\ldots A_N}^{B}= \sum_{C_0}\cdots \sum_{C_N}\,  Z_{A_1}^{C_1}\cdots Z_{A_N}^{C_N} (Z^{-1})_{C_0}^B\, \C_{C_1\ldots C_N}^{C_0}\, ,
\een
where we note that all summations are finite because $Z$ is upper triangular. Combining eqs.\eqref{recursionintro} and \eqref{tildetrafo}, we immediately see that the redefined OPE coefficients now satisfy the recursion formula (suppressing spacetime arguments) 
\ben\label{transformedrecur}
\begin{split}
&\partial_{g}\widehat\C_{A_1\ldots A_N}^{B}= \partial_{g} \left(Z_{A_1}^{C_1}\cdots Z_{A_N}^{C_N} (Z^{-1})_{C_0}^B\, \C_{C_1\ldots C_N}^{C_0} \right)  \\
=&-Z_{A_1}^{C_1}\cdots Z_{A_N}^{C_N} (Z^{-1})_{C_0}^B\, \int_y \Big[  \C_{\INT C_1\ldots C_N}^{C_0} -\sum_{i=1}^N\sum_{[D]\leq [C_i]} \C_{\INT C_i}^D \C_{C_1\ldots D\ldots C_N}^{C_0}- \sum_{[D]< [C_0]} \C_{C_1\ldots C_N}^{D} \C_{\INT D}^{C_0} \Big]\\
& -Z_{A_1}^{C_1}\cdots Z_{A_N}^{C_N} (Z^{-1})_{C_0}^B \Big[  \sum_{i=1}^N \Gamma_{ C_i}^D \C_{C_1\ldots D\ldots C_N}^{C_0} - \C_{C_1\ldots C_N}^{D} \Gamma_{ D}^{C_0}  \Big]
\end{split}
\een
where the objects $\Gamma_{A}^{B}$ are defined as elements of the matrix $\Gamma$,
\ben\label{GammaZ}
\Gamma:= - \, Z^{-1}\, \partial_{g} Z \, .
\een
We would like to make a specific choice of the mixing matrix $Z$ in order to cancel the contribution to the integral \eqref{transformedrecur} coming from large $|y|$ (infra-red region). For that purpose, we define 
\ben\label{Gammadef}
\Gamma_{A}^B :=\begin{cases}
\int_{|y|>L} \C_{\INT A}^B (y)\,  \d^4 y  & \text{ for }[A]\geq [B]\\
 0  & \text{ for }[A]< [B]
\end{cases}
\een
for some  $L>0$. (Note that $\Gamma$ depends on ${g}$.) The solution to eq.\eqref{GammaZ} can then formally be written as
\ben\label{redef2}
Z(g)= \mathcal{P}\exp -\int_0^g   \Gamma(g')\,  \d g'  \, ,
\een
where $\mathcal{P} \exp$ denotes the ``path ordered exponential". 

Combining this definition of $Z$ with \eqref{transformedrecur} and with the associativity property \eqref{associntro} and choosing $L>\max_i|x_i-x_N|$, we can rewrite the recursion formula for the \emph{new} OPE coefficients $\widehat{\C}_{A_1\ldots A_N}^B$ as
\ben\label{recurnew1}
\begin{split}
\partial_{g}\widehat\C_{A_1\ldots A_N}^{B}(x_1,\ldots, x_N) =-& \int\limits_{|y-x_N|\leq L} \d^4y \sum_{[E]\leq 4} (Z^{-1})_\INT^E \,  \Big[  \widehat\C_{E A_1\ldots A_N}^{B}(y,x_1,\ldots, x_N)\\
&-\sum_{i=1}^N\sum_{[D]\leq [A_i]} \widehat\C_{E A_i}^D(y,x_i) \widehat\C_{A_1\ldots D\ldots A_N}^{B}(x_1,\ldots, x_N)\\
&- \sum_{[D]< [B]} \widehat\C_{A_1\ldots A_N}^{D}(x_1,\ldots, x_N) \widehat\C_{E D}^{B}(y,x_N) \Big]\, .
\end{split}
\een
Here the idea behind our redefinition \eqref{Gammadef} becomes apparent: We have arrived at a modified recursion formula which includes only integrals over a region of finite volume. The contributions to the integrals from $|y-x_N| > L$ have been cancelled precisely by the terms coming from the field redefinition (using also the associativity theorem \ref{thmassoc}). 

We would finally like to tidy up the factors $ (Z^{-1})_\INT^E $ in front of the OPE coefficients in \eqref{recurnew1} by a redefinition of our coupling constant $g$. In particular, we would like to choose this redefinition of $g$ in such a way that the formula \eqref{recurnew1} has a simple and well defined limit $m^2\to 0$. The following lemma allows us to understand the small mass behaviour of the mixing matrix $Z$:
\begin{lemma}\label{mixinglimit}
The mixing matrix behaves as
\ben
\lim_{m^2\to 0}[ Z_\INT^\INT\cdot (Z^{-1})_\INT^A ]= \delta_\INT^A +K\cdot \delta_{(\varphi\partial^2\varphi)}^A
\een
for some formal power series $K({g})$. 
\end{lemma}
\begin{proof}
We establish this lemma by analysing the small mass behaviour of the OPE coefficients appearing in the matrix elements $Z_\INT^A$. More precisely, we will prove that
\ben
\int\limits_{|x|>L} (\C_r)_{\INT\INT}^A(x) = K_{A}\cdot \left[\log(L^2  m^2)\right]^{r+1} + O\Big(\left[\log(L^2 m^2)\right]^{r}\Big)  \text{ for }A: \O_A\in\{\varphi^4,\varphi\partial^2\varphi\} \label{bigA}
\een
\ben
\int\limits_{|x|>L} (\C_r)_{\INT\INT}^A(x) = O\Big(\left[\log(L^2m^2)\right]^{r}\Big)  \text{ for } [A]\leq 4\, , \O_A\notin\{\varphi^4,\varphi\partial^2\varphi\}  \label{smallA}
\een
for some constants $K_{A}$ which depend on the perturbation order, and where $K_{A}\neq 0$ for $\O_{A}=\varphi^{4}$. These equations then imply that the rescaled matrix $Z_\INT^A/Z_\INT^\INT$ vanishes in the limit $m^2\to 0$ unless $\O_A=\varphi^4$ or $\O_A=\varphi\partial^2\varphi$, which upon inversion of this matrix leads directly to the lemma. 

To prove these statements, we are going to proceed inductively. Using eq.\eqref{OPEpf}, one checks \eqref{bigA} and \eqref{smallA} for the free theory by straightforward computation.  For the induction step we make use of our original recursion formula \eqref{recursionintro}. Using the associativity property \eqref{associntro}, we can rewrite the recursion formula in the useful form
\ben\label{recurIR}
\begin{split}
\int_{|x|>L}%^{|x|=\Lambda} 
 (\C_{r+1})_{\INT \INT}^A(x) &= \int_{|x|>L}%^{|x|=\Lambda}  
  \int\limits_{y\in\Omega_1} \Big[\sum_{[D]\leq 4} (\C_{r_1})_{\INT \INT}^D(x-y) (\C_{r_2})_{D \INT}^A(x) \\
&\qquad +\sum_{[D]<[A]} (\C_{r_1})_{\INT\INT}^D(x) (\C_{r_2})_{\INT D}^A(y) - \sum_{[D]> 4} (\C_{r_1})_{\INT \INT}^D(y) (\C_{r_2})_{\INT D}^A(x)  \Big]\\
&+ \int_{|x|>L}%^{|x|=\Lambda}   
\int\limits_{y\in\Omega_2} \Big[  \sum_{[D]\leq 4} (\C_{r_1})_{\INT \INT}^D(y) (\C_{r_2})_{\INT D}^A(x)  \\
& \qquad+ \sum_{[D]\leq 4} (\C_{r_1})_{\INT \INT}^D(x-y) (\C_{r_2})_{D \INT}^A(x)- \sum_{[D]\geq [A]} (\C_{r_1})_{\INT\INT}^D(x) (\C_{r_2})_{\INT D}^A(y) \Big]\\
&+ \int_{|x|>L}%^{|x|=\Lambda}   
\int\limits_{y\in\Omega_3} \Big[  \sum_{[D]\leq 4} (\C_{r_1})_{\INT \INT}^D(y) (\C_{r_2})_{\INT D}^A(x)   \\
& \qquad+ \sum_{[D]< [A]} (\C_{r_1})_{\INT\INT}^D(x) (\C_{r_2})_{\INT D}^A(y)- \sum_{[D]> 4} (\C_{r_1})_{\INT \INT}^D(x-y) (\C_{r_2})_{D \INT}^A(x) \Big]\\
\end{split}
\een
where the regions $\Omega_i\subset \mathbb{R}^4$ are defined as
\begin{align}
\Omega_1&:=\{  y\in\mathbb{R}^4 \, :\, |y|(1+\varepsilon) < |x| \} \\
\Omega_2&:=\{  y\in\mathbb{R}^4 \, :\, |y|>  |x| (1+\varepsilon) \} \\
\Omega_3&:= \mathbb{R}^4 \setminus (\Omega_1\cup\Omega_2)
\end{align}
for some $\varepsilon>0$. Note that the infinite sums in eq.\eqref{recurIR} are absolutely convergent by our theorem \ref{thmassoc}. Considering first the case $A=\INT$ and focusing on the contributions of leading order in $\log(m^2)$, we are left with
\ben\label{recurIRleading}
\begin{split}
\int_{|x|>L}  (\C_{r+1})_{\INT \INT}^\INT(x) &= \int_{|x|>L}   \int\limits_{y\in\Omega_1}  (\C_{r_1})_{\INT \INT}^\INT(x-y) (\C_{r_2})_{\INT \INT}^\INT(x)  \\
&+ \int_{|x|>L}   \int\limits_{y\in\Omega_2} (\C_{r_1})_{\INT \INT}^\INT(x-y) (\C_{r_2})_{\INT \INT}^\INT(x)\\
&+ \int_{|x|>L}   \int\limits_{y\in\Omega_3}  (\C_{r_1})_{\INT \INT}^\INT(y) (\C_{r_2})_{\INT \INT}^\INT(x)    + O\Big(\left[\log(L^2m^2)\right]^{r+1}\Big) \, .
\end{split}
\een
Here we used the induction hypotheses, eqs.\eqref{bigA} and \eqref{smallA}, in order to estimate the small $m$ behaviour of the coefficients $\C_{\INT\INT}^A$ and we used the bounds
\ben\label{smallmOPE1}
\Big|  \int_{|x|>L}  (\C_r)_{AB}^C(x)\cdot |x|^{[A]+[B]-[C]-4} \,\d^4 x \Big| \leq  O\Big(\left[\log(L^2m^2)\right]^{r+1}\Big) 
\een
for $[A]+[B]-[C]\geq 4$, and 
\ben\label{smallmOPE2}
\begin{split}
\Big|  \int_{|x|=L}^{|x|=\Lambda} (\C_r)_{AB}^C(x) \,\d^4 x \Big| &\leq  \Lambda^{4+[C]-[A]-[B]}   O\Big(\left[\log(\Lambda^2m^2)\right]^{r}\Big)
\end{split}
\een
for $[A]+[B]-[C]< 4$ in order to estimate the other OPE coefficients appearing in \eqref{recurIR}. These bounds can be established inductively: They are easily verified at zeroth order using eq.\eqref{OPEpf}, and, using the recursion formula in the form \eqref{recurIR}, one picks up an additional power of $\log(L^2m^2)$ with every iteration. Furthermore, we also used the fact that $\C_{\INT (\varphi\partial^{2}\varphi)}^{\INT}=O(m^{2})$ to obtain \eqref{recurIRleading}, which can also be shown inductively using $\partial^2\Delta(x)=-m^2\Delta(x)+\delta(x)$. 
Applying the induction hypothesis \eqref{bigA} in order to estimate the remaining terms in eq.\eqref{recurIRleading}, we see that indeed we obtain a non-vanishing contribution of the order $[\log(L^2m^2)]^{r+2}$, as claimed.

The other estimate stated in eqs.\eqref{bigA} follows directly from \eqref{smallmOPE1}. Regarding \eqref{smallA}, we note that the zeroth order OPE coefficient $(\C_{0})_{\INT\INT}^{(\partial\varphi)^{2}}$ vanishes. Using this in the recursion formula \eqref{recurIR}, one can verify \eqref{smallA} by induction. For the integral over the coefficients  $\C_{\INT\INT}^{A}$ with $[A]<4$ one can even check that the limit $m^{2}\to 0$ is finite at zeroth order, so \eqref{smallA} certainly holds at higher orders by iteration of the recursion formula. 
\end{proof}
Combining lemma \ref{mixinglimit} with a redefinition of the coupling constant
\ben\label{gredef}
\partial_{\hat{g}} = Z_\INT^\INT \, \partial_g\, ,
\een
we arrive at the following
\begin{prop}\label{prop}
There exists a field redefinition,  eq.\eqref{Zhat}, and a redefinition of the coupling constant, eq.\eqref{gredef}, such that the recursion formula for the redefined OPE coefficients has a well defined massless limit. For $m^2=0$ this formula reads 
\ben\label{recurnew}
\begin{split}
\partial_{\hat{g}}\widehat\C_{A_1\ldots A_N}^{B}(x_1,\ldots, x_N) =-& \int\limits_{|y-x_N|\leq L} \d^4y \, \Big[  \widehat\C_{\INT A_1\ldots A_N}^{B}(y,x_1,\ldots, x_N)\\
&-\sum_{i=1}^N\sum_{[D]\leq [A_i]} \widehat\C_{\INT A_i}^D(y,x_i) \widehat\C_{A_1\ldots D\ldots A_N}^{B}(x_1,\ldots, x_N)\\
&- \sum_{[D]< [B]} \widehat\C_{A_1\ldots A_N}^{D}(x_1,\ldots, x_N) \widehat\C_{\INT D}^{B}(y,x_N) \Big]
\end{split}
\een
for any $L>\max_i|x_i-x_N|$ in the sense of formal power series in $\hat{g}$. 
\end{prop}
\begin{proof}
Using lemma \ref{mixinglimit} in eq.\eqref{recurnew1}, it only remains to show that the contribution from the sum over $E$ with $\O_{E}=\varphi\partial^{2}\varphi$ vanishes. This is achieved by induction. Using eq.\eqref{OPEpf}, which also holds for the new coefficients since $(\widehat{\C}_0)=(\C_0)$, and using also the fact that $(\partial^{2}+m^2)\Delta(x) =\delta(x) $, one verifies that the term in question, i.e.
\ben\label{delsqterm}
\begin{split}
   \widehat\C_{(\varphi\partial^{2}\varphi) A_1\ldots A_N}^{B}(y,x_1,\ldots, x_N)&-\sum_{i=1}^N\sum_{[D]\leq [A_i]} \widehat\C_{(\varphi\partial^{2}\varphi) A_i}^D(y,x_i) \widehat\C_{A_1\ldots D\ldots A_N}^{B}(x_1,\ldots, x_N)\\
   &- \sum_{[D]< [B]} \widehat\C_{A_1\ldots A_N}^{D}(x_1,\ldots, x_N) \widehat\C_{(\varphi\partial^{2}\varphi) D}^{B}(y,x_N) \, ,
\end{split}
\een
vanishes at zeroth order in the limit $m^2\to 0$. To show that this term also vanishes to all orders in perturbation theory, we write the corresponding recursion formula  in the form\footnote{In the derivation of \eqref{writeEdernice} we have exchanged the coefficient $\C_{\varphi^4 (\varphi\partial^2 \varphi)}^C$ for the coefficient $\C_{ (\varphi\partial^2 \varphi)\varphi^4}^C$. This is a non-trivial procedure in the case where $C\in\{(\partial\varphi)^2,(\varphi\partial^2 \varphi) \}$, since in that case these coefficients do not actually coincide. However, we note that in \eqref{writeEdernice} they multiply vanishing contributions of the form $\mathbf{T}_{C(y)} [\C_{A_1\ldots A_N}^B(x_1,\ldots, x_N)]$, so exchanging the order of the indices is indeed justified.   }
\ben\label{writeEdernice}
\begin{split}
\partial_{\hat{g}}[\text{\eqref{delsqterm}}] =& -\int\limits_{|z-x_N|\leq L} \d^4 z  \sum_{[E]\leq 4} (Z^{-1})_\INT^E \,   \Big[\mathbf{T}_{\varphi\partial^{2}\varphi(y)} [\C_{E A_1\ldots A_N}^B(z,x_1,\ldots, x_N)] \\
&- \sum_{i=1}^N \sum_{[C]\leq [A_i]} \mathbf{T}_{\varphi\partial^{2}\varphi(y)}[ \C_{E A_i}^C(z,x_i)  \cdot \C_{A_1\ldots C\ldots  A_N}^B(x_1,\ldots, x_N)] \\
&- \sum_{[C]< [B]} \mathbf{T}_{\varphi\partial^{2}\varphi(y)} [\C_{A_1\ldots   A_N}^C(x_1,\ldots, x_N) \,  \C_{E C}^B(z,x_N)]  \Big]
\end{split}
\een 
where we defined the operator 
\ben
\text{\eqref{delsqterm}} =: \mathbf{T}_{\varphi\partial^{2}\varphi(y)} [\C_{A_1\ldots A_N}^B(x_1,\ldots, x_N)] 
\een
which acts on products of OPE coefficients by the Leibniz rule. Thus, assuming the expression \eqref{delsqterm} vanishes up to perturbation order $r$, it follows from eq.\eqref{writeEdernice} that it will also vanish at order $r+1$. This closes the induction and proves eq.\eqref{recurnew}.
\end{proof}

\vspace{.5cm}

\noindent One may view eq.\eqref{recurnew} as providing a definition for the OPE coefficients of massless $\varphi^4$-theory: We simply define the OPE coefficients of the massless theory 
to be the obvious ones in the free theory [i.e. setting $m^2=0$ in eq.\eqref{OPEpf}], and then define the higher orders via eq. \eqref{recurnew}. The OPE coefficients of this massless theory are then defined as a formal series in $\hat g$.

\subsection{OPE associativity for massless fields}

Defining the OPE coefficients of the massless theory via proposition \ref{prop} as discussed in the previous subsection, the theorem is again that the resulting definition is consistent, i.e. does not lead to UV-divergences at any order and satisfies the associativity condition at any order 
in $\hat g$:
\begin{thm} \label{thmmassless}
The OPE coefficients of massless Euclidean $\varphi^{4}$-theory, as defined recursively through eqs. \eqref{OPEpf} and \eqref{recurnew}, satisfy the associativity property \eqref{associntro} on the domain \eqref{domain} to any order in perturbation theory. 
\end{thm}

\begin{proof}[Sketch of proof:] With the modified recursion formula \eqref{recurnew} at hand, we can copy our strategy from the massive case in order to prove associativity of the OPE also for massless fields. As the differences in the proof are minor, we refrain form repeating the lengthy calculations here. Instead, we only point out the main adjustments that have to be made. 

Most importantly,  one has to adapt the induction hypothesis \eqref{hypothesis1} to the massless case by replacing factors of $1/m$ by the length scale $L$ appearing in the modified recursion formula. The induction step remains largely the same. Here we can simply take the limit $m\to 0$ in the bound \eqref{freefinalbound}, which forces us to choose $\delta_u=0$. The only essential difference appears in the estimation of the recursion integral \eqref{remr+1} over the large distance region $\Omega_{IR}$. With the modified recursion formula, this region now has a  {cutoff} $L$. The estimates \eqref{intest2} and \eqref{intest3} are therefore replaced by 
\ben
 \int_{\Omega^v_{IR}}  \frac{\d^4 y}{|y-x_w|^{4}} \leq  (2\pi)^{2} \left(\frac{L}{\min_{i\in\n(w)}|x_w-x_i|}\right)^\delta
\een
for any $\delta>0$. Taking into account these adjustments, the proof carries over from the massive case without further complications. 
\end{proof}

\section{Conclusions}\label{sec:con}

In this paper we have shown that the operator product expansion in Euclidean $\varphi^4_4$-theory satisfies an associativity condition that was originally conjectured in~\cite{Hollands:2008vq}. The model is therefore the first non-trivial example of a quantum field theory satisfying all the axioms of the framework proposed in~\cite{Hollands:2008vq} (see also sec.~\ref{sec:AX} of the present paper). Further, all results derived in that paper which were based on the assumption of associativity, i.e. the coherence theorem, the formulation of perturbation theory in terms of Hochschild cohomology and the relation to vertex operator algebras, are now established within Euclidean $\varphi^4_4$-theory as a corollary of the associativity theorem. 
  As a side result of the present paper, we have also shown how to adapt the recursion formula for OPE coefficients, which was originally only derived for massive fields, to the massless case.

The method of proof followed in the present paper can be straightforwardly adapted to other self-interacting Euclidean quantum field theory models. Hence, the associativity condition should also hold for example in the Euclidean Thirring- and the  Gross-Neveu model. 

Generalisations of our result in various directions would be of interest, e.g. to theories with gauge symmetry or to models on curved background manifolds. In particular, it may be possible to generalise the finite volume recursion formula \eqref{recurnew} to (Riemannian-) curved manifolds if the scale $L$ is chosen small enough such that one can use Riemann normal coordinates to study the $y$-integral. By far the most exciting potential application of our results is that they may help to give a non-perturbative definition of quantum field theory in the sense outlined in section \ref{sec:AX}.

\vspace{.5cm}

\paragraph{Acknowledgements:}
Our research was supported by ERC starting grant QC\& C 259562.  SH is grateful to the Kavli Institute for Theoretical Physics, UCSB, for hospitality and financial support during the program ``Quantum Gravity Foundations: UV to IR", where some of the results in this paper were presented. 

\vspace{.5cm}

\appendix

\section{Zeroth order bounds}

Below we derive explicit bounds on zeroth order OPE coefficients which are used to verify the inductive bound \eqref{hypothesis1} at the induction start $r=0$. More specifically, we first estimate Taylor expansions of the Euclidean propagator in section \ref{taylorapp} and then apply the resulting bound in section \ref{lemm2app} in order to verify the estimate claimed in lemma \ref{lembound} above.

\subsection{Taylor expansions of the propagator}\label{taylorapp}

For free quantum fields, the operator product expansion is closely related to the Taylor expansion of the propagator. As we have seen for example in lemma \ref{lemHf}, the same holds true for the contractions of OPE coefficients $\Pro_0(T)$ considered in this paper. It should therefore not come as a surprise that a central ingredient in our derivation of the upper bounds on $|\Pro_0(T)|$ are bounds on Taylor expansions of the propagator. More precisely, we make use of the following lemma [recall that by $\Delta(x)$ we denote the Euclidean propagator, eq.\eqref{propagator}]:
\begin{lemma}\label{lemTaylor}
For any $\varepsilon\in(0,\frac{1}{8r}]$, any $\delta\in[0,1]$, any $w\in\mathbb{N}^4$ and any $(d_1,\ldots, d_r)\in\mathbb{N}^{r}$, one has
\ben\label{tayloreq}
\begin{split}
&\Big| \sum_{|v_1|=d_1}\cdots \sum_{|v_r|=d_r}  \frac{x_1^{v_1}}{v_1!}\partial_{y}^{v_1}  \cdots \frac{x_r^{v_r}}{ v_r!}  \, \partial_{y}^{v_r} \partial^w_{y} \Delta(y) \Big| \\
&\qquad \leq (|w|+\delta)!\, \frac{(|x_1|/\varepsilon^2)^{d_1}\cdots (|x_{r-1}|/\varepsilon^2)^{d_{r-1}} \, [(1+\varepsilon)|x_r|]^{d_r} }{\varepsilon^{4+2|w|+2\delta}\cdot |y|^{2+|w|+\sum d_i+\delta} \, m^\delta} \, .
\end{split}
\een
\end{lemma}
\begin{proof}
Our strategy is to pull the modulus into the summations as follows,
\ben\label{propbound1}
\begin{split}
&\Big| \sum_{|v_1|=d_1}\cdots \sum_{|v_r|=d_r}  \frac{x_1^{v_1}}{v_1!}\partial_{y}^{v_1}  \cdots \frac{x_r^{v_r}}{ v_r!}  \, \partial^w_{y} \partial_{y}^{v_r} \Delta(y) \Big|  \\
& \leq  \sum_{|v_1|=d_1}\cdots \sum_{|v_{r-1}|=d_{r-1}} \Big|  \frac{x_1^{v_1}}{v_1!} \cdots \frac{x_{r-1}^{v_{r-1}}}{ v_{r-1}!} \Big| \cdot \Big|\sum_{|v_r|=d_r}  \frac{x_r^{v_r}}{v_r!}\partial_{y}^{v_1+\ldots+v_r+w} \Delta(y) \Big| \\
& =  \sum_{|v_1|=d_1}\cdots \sum_{|v_{r-1}|=d_{r-1}} \Big|  \frac{x_1^{v_1}}{v_1!} \cdots \frac{x_{r-1}^{v_{r-1}}}{ v_{r-1}!} \Big| \cdot \Big| \frac{\partial_\tau^{d_r}}{d_r!}  \partial^{v_1+\ldots+v_{r-1}+w} \Delta(y+\tau x_r) \Big|_{\tau=0}
\end{split}
\een
To bound the last factor on the r.h.s., we write it as a contour integral:
\ben\label{cauchy}
\frac{\partial_\tau^{d_r}}{d_r!}  \partial^{u} \Delta(y+\tau x_r)\Big|_{\tau=0} =\frac{1}{2\pi i} \oint_\gamma \frac{ \partial^{u} \Delta(y+z x_r)}{z^{d_r+1}}\, \d z
\een
Here we use the shorthand $u:=v_1+\ldots+v_{r-1}+w$, and $\gamma$ is any circle around the origin in the complex such that $\partial^{u} \Delta(y+z x_r)$ is holomorphic on the closed disk bounded by this circle. Since the propagator has a pole at the origin, $\gamma$ is restricted to circles with radius $R<|y|/|x_r|$. We therefore write
\ben
R=\frac{|y|}{|x_r| } \cdot \frac{1}{1+\varepsilon}
\een
where $\varepsilon>0$ is arbitrary. From eq.\eqref{cauchy} we then obtain the bound
\ben\label{cauchybound}
\Big|\frac{\partial_\tau^{d_r}}{d_r!}  \partial^{u} \Delta(y+\tau x_r)\Big|_{\tau=0} \leq \frac{\sup_{z\in \gamma} |\partial^{u} \Delta(y+z x_r)| }{R^{d_r}} \, .
\een
In order to estimate the numerator, we write the propagator explicitly as
\ben
\partial^{u} \Delta(x) =      \frac{1}{16\pi^{2}}  \partial^u  \int_{0}^{\infty}\d t \exp\left(-t m^2-\frac{x^{2}}{4t}\right)\,t^{-2}  \, .
\een
Using the inequality~\cite[eq.(56)]{Hollands:2011gf}
\ben
|\partial^u \e^{-\frac{x^{2}}{4t}} | \leq c \,  t^{-|u|/2}\, \sqrt{|u|!} \, 2^{-|u|/2} \, \e^{-\frac{x^{2}}{8t}} \, , \qquad c<2 
\een
we obtain the bound 
\ben
\begin{split}
|\partial^{u} \Delta(x) | &\leq      \frac{c\, \sqrt{|u|!} \, }{2^{|u|/2}16\pi^{2} }   \int_{0}^{\infty}\d t   \,  t^{-|u|/2-2} \, \exp\left(-t m^2-\frac{x^{2}}{8t}\right) \\
&\leq \frac{ \sqrt{|u|!} }{2^{|u|/2} m^{\delta}16\pi^{2} }   \int_{0}^{\infty}\d t   \,  t^{-(|u|+\delta)/2-2}   \, \exp\left(-\frac{x^{2}}{8t}\right)   \leq \left(\frac{4}{x^2}\right)^{(|u|+\delta)/2+1} \cdot  \frac{(|u|+\delta)! }{m^{\delta}}\, .
\end{split}
\een
Substituting this estimate in \eqref{cauchybound} and noting that $\sup_{z\in\gamma}(1/|y+z x_r|)= (1+\varepsilon)/(\varepsilon |y|)$, we arrive at the bound
\ben
\Big|\frac{\partial_\tau^{d_r}}{d_r!}  \partial^{u} \Delta(y+\tau x_r)\Big|_{\tau=0} \leq \frac{(|u|+\delta)!}{m^\delta}  \, \left((1+\varepsilon)\cdot \frac{|x_r|}{|y|}\right)^{d_r} \, \left(\frac{2(1+\varepsilon)}{\varepsilon|y|}\right)^{|u|+\delta+2}
\een
Combining this bound with the inequality
\ben
\sum_{|v_1|=d_1}\cdots\sum_{|v_{r-1}|=d_{r-1}} \Big| \frac{x_1^{v_1}}{v_1!}\cdots \frac{x_{r-1}^{v_{r-1}}}{v_{r-1}!}  \Big| \leq \frac{(2(r-1)|x_1|)^{d_1}\cdots (2(r-1)|x_{r-1}|)^{d_{r-1}}}{(d_1+\ldots +d_{r-1})!}
\een
and choosing $\varepsilon\leq \frac{1}{8r}$ we finally arrive at the claimed bound \eqref{tayloreq}, which  finishes the proof of the lemma.
\end{proof}

\subsection{Proof of Lemma \ref{lembound}}\label{lemm2app}

We want to derive a bound on the matrix elements $M_{\pi}$ defined in eq.\eqref{Mweight}, where $\pi=[(v,i)(w,j)]\in\sigma$ for some perfect matching $\sigma\in\mathfrak{M}(\leaf\cup\root)$. Let us first assume that $v,w\neq \root$. Further, let us write explicitly $\an(v)\setminus\an(w)=(u_1,\ldots, u_a)$  and $\an(w)\setminus\an(v)=(s_1,\ldots, s_{b})$, where we use the convention that $u_i$ is closer to the leaves than $u_{i+1}$, and the same for $s_{i}$.  We can then write equation \eqref{Mweight} explicitly as 
\ben\label{Mij0}
\begin{split}
M_{\pi}(\vec{d})&=\sum_{|\alpha_{u_1}|\leq d_{\pi}^{u_1}-|\alpha_{v,i}|}\cdots \sum_{|\alpha_{u_a}|\leq d_{\pi}^{u_a}-|\alpha_{v,i}|}  \sum_{|\alpha_{s_1}|\leq d_{\pi}^{s_1}-|\alpha_{w,j}|}\cdots \sum_{|\alpha_{s_{b}}|\leq d_{\pi}^{s_{b}}-|\alpha_{w,j}|}   \\
&\qquad \times  \frac{(x_v-x_{u_1})^{\alpha_{u_1}}}{\alpha_{u_1}!}\frac{(x_{u_1}-x_{u_{2}})^{\alpha_{u_{2}}-\alpha_{u_1}}}{(\alpha_{u_{2}}-\alpha_{u_1})!}\cdots \frac{(x_{u_{a-1}}-x_{u_{a}})^{\alpha_{u_{a}}-\alpha_{u_{a-1}}}}{(\alpha_{u_{a}}-\alpha_{u_{a-1}})!}  \\
&\qquad\times \frac{(x_w-x_{s_1})^{\alpha_{s_1}}}{\alpha_{s_1}!}\cdots \frac{(x_{s_{b-1}}-x_{s_{b}})^{\alpha_{s_{b}}-\alpha_{s_{b-1}}}}{(\alpha_{s_{b}}-\alpha_{s_{b-1}})!}  \quad  \partial_{x_{u_a}}^{\alpha_{u_a}+\alpha_{v,i}}    \partial_{x_{s_{b}}}^{\alpha_{s_{b}}+\alpha_{w,j}}   \Delta(x_{u_a}- x_{s_{b}}) 
\end{split}
\een
Using the bound \eqref{tayloreq} from lemma~\ref{lemTaylor}, we obtain 
\ben\label{Mij0bd}
\begin{split}
&|M_{\pi}| \leq \frac{|x_v-x_{u_1}|^{d_{\pi}^{u_1}} \theta(d_{\pi}^{u_1}-|\alpha_{v,i}|) \cdots |x_{s_{b-1}}-x_{s_{b}}|^{d_{\pi}^{s_{b}}-d_{\pi}^{s_{b-1}}} \theta(d_\pi^{s_{b}}-d_\pi^{s_{b-1}})  }{m^\delta\,  |x_{u_a}-x_{s_{b}}|^{2+ d_{\pi}^{s_{b}}+ d_{\pi}^{u_{a}}+\delta   } \cdot |x_v-x_{u_1}|^{|\alpha_{v,i}|} \cdot |x_w-x_{s_1}|^{|\alpha_{w,j}|}  } \\
&\times   (|\alpha_{v,i}|+|\alpha_{w,j}|+\delta)! \cdot \varepsilon^{-2(d_{\pi}^{u_{a-1}}+d_{\pi}^{u_b}+2+|\alpha_{v,i}|+|\alpha_{w,j}|+\delta)}\, (1+\varepsilon)^{d_{\pi}^{u_a}-d_{\pi}^{u_{a-1}}}  \\
& \leq  \frac{(|\alpha_{v,i}|+|\alpha_{w,j}|+\delta)!}{(\varepsilon^2 |x_v-x_{u_1}|)^{|\alpha_{v,i}|+1} \cdot (\varepsilon^2|x_w-x_{s_1}|)^{|\alpha_{w,j}|+1}  \ m^{\delta} (\varepsilon^2|x_{u_a}-x_{s_{b}}|)^{\delta} }\\
&\times  \theta(d_\pi^{u_{a}}-d_\pi^{u_{a-1}}) [\xi_{u_{a}} (1+\varepsilon)]^{1+d_{\pi}^{u_{a}}}  \prod_{i=1}^{a-1}  \theta(d_\pi^{u_{i}}-d_\pi^{u_{i-1}}) (\frac{\xi_{u_{i}}}{\varepsilon^2})^{1+d_{\pi}^{u_{i}}}  \prod_{j=1}^{b}  \theta(d_\pi^{s_{j}}-d_\pi^{s_{j-1}}) (\frac{\xi_{s_{j}}}{\varepsilon^2})^{1+d_{\pi}^{s_{j}}} 
\end{split}
\een
for any $\varepsilon\in(0,1/8(a+b)]$. Since $8(a+b)\leq 8|\interoot|\leq 8|\leaf| \leq 8\dext\leq  2^{\dext+3}$, we can always choose $\varepsilon\in(0,1/2^{\dext+3}]$, which already establishes lemma \ref{lembound} for the case where $\{u_1,\ldots, u_a, s_1,\ldots, s_b\} \cap b(T) = \emptyset$.  

Thus, assume now that one of the vertices $u_{i}$ is in $ b(T)$. In this case, we note that also the vertices $u_{i+1},\ldots, u_{a}$ belong to $b(T)$ on account of being ancestors of $u_i$. Further,  also know that none of the vertices $(s_1,\ldots, s_{b})$ belong to $b$, since none of them is an ancestor of $u_i$ by definition.  If $u_{a-1}\in b(T)$, then it is easy to see that the sum over $\alpha_{u_{a-1}}$ simply yields a Kronecker delta $\delta_{d_{\pi}^{u_a},d_{\pi}^{u_{a-1}}}$ since by definition all vertices in $b$ have the same associated coordinate, i.e. $x_{u_{a}}=x_{u_{a-1}}$ in that case. We can repeat the procedure with the line $u_{a-2}$ if it is in $b(T)$ as well. Renaming summation indices, we can therefore reduce \eqref{Mij0} to a form where \emph{only} the index $u_{a}$ corresponds to a line in $b(T)$. Thus, we see that vertices in $b(T)$ come with factors of $(1+\varepsilon)$ instead of $1/\varepsilon^2$, which is also consistent with the bound \eqref{Mijbound} in lemma \ref{lembound}.

Next we come to the case $w=\root$. In this case $\an(w)=\emptyset$, so $M_{\pi}$ is simply
\ben\label{Mij0B}
\begin{split}
|M_{\pi}|&=\Big|\sum_{|\alpha_{u_1}|\leq d_{\pi}^{u_1}-|\alpha_{v,i}|}\cdots \sum_{|\alpha_{u_a}|\leq d_{\pi}^{u_a}-|\alpha_{v,i}|} \\
&\times    \frac{(x_v-x_{u_1})^{\alpha_{u_1}}}{\alpha_{u_1}!}\frac{(x_{u_1}-x_{u_{2}})^{\alpha_{u_{2}}-\alpha_{u_1}}}{(\alpha_{u_{2}}-\alpha_{u_1})!}\cdots \frac{(x_{u_{a-1}}-x_{u_{a}})^{\alpha_{u_{a}}-\alpha_{u_{a-1}}}}{(\alpha_{u_{a}}-\alpha_{u_{a-1}})!}      \, \frac{(x_{u_a}-x_\root)^{\alpha_{w,j}-\alpha_{v,i}-\alpha_{u_a}}}{(\alpha_{w,j}-\alpha_{v,i}-\alpha_{u_a})!} \Big|\\
&\leq \frac{|x_v-x_{u_1}|^{d_{\pi}^{u_1}-|\alpha_{v,i}|} \theta(d_\pi^{u_1}-|\alpha_{v,i}|) }{(d_{\pi}^{u_1}-|\alpha_{v,i}|)!} %\cdot \frac{|x_{e_1}-x_{e_2}|^{d_{ij}^{e_2}-d_{ij}^{e_1}}}{(d_{ij}^{e_2}-d_{ij}^{e_1})!} 
\cdots \frac{|x_{u_a}-x_{\root}|^{|\alpha_{w,j}|-d_{\pi}^{u_a}  }  \theta(|\alpha_{w,j}|-d_{\pi}^{u_a}) }{(|\alpha_{w,j}|-d_{\pi}^{u_a})!} \\
& \leq \frac{|x_{u_a}-x_\root|^{|\alpha_{w,j}|+1} }{|x_{u_1}-x_v|^{|\alpha_{w,i}|+1}}  \, \prod_{e\in \interoot}  \theta(d_\pi^{e}-d_\pi^{\c(e)}) \xi_e^{d_{\pi}^e+1}
\end{split}
\een
This is consistent with the claimed bound \eqref{Mijbound}, and therefore finishes the proof of lemma \ref{lembound}. \hfill \qedsymbol

\bibliographystyle{utphys}
\bibliography{assoc}

\end{document}